\documentclass[screen, nonacm]{acmart}
\usepackage[UKenglish]{babel}

\acmJournal{TEAC}

\newtheorem{observation}{Observation}

\usepackage{algorithm}
\usepackage[noend]{algorithmic}
\usepackage{bm}
\usepackage{subfigure}

\newcommand{\bids}{\mathcal{B}}
\newcommand{\learnt}{\mathcal{L}}

\newcommand{\R}{\mathbb{R}}
\newcommand{\Rn}{\mathbb{R}^n}
\newcommand{\Z}{\mathbb{Z}}
\newcommand{\Zn}{\mathbb{Z}^n}
\newcommand{\uR}{\underline{\mathbb{R}}}
\newcommand{\qoracle}{\mathcal{Q}_\bids}
\newcommand{\sq}{\textnormal{SQ}}
\DeclareMathOperator{\conv}{\textnormal{conv}}

\DeclareMathOperator{\dom}{\textnormal{dom}}
\DeclareMathOperator{\supp}{\textnormal{supp}}

\DeclareMathOperator*{\argmax}{\textnormal{arg\,max}}
\renewcommand{\vec}[1]{\bm{#1}}
\newcommand{\bid}{\vec{b}}
\newcommand{\eb}{\vec{e}}
\newcommand{\pb}{\vec{p}}
\newcommand{\qb}{\vec{q}}
\newcommand{\vb}{\vec{v}}
\newcommand{\xb}{\vec{x}}
\newcommand{\yb}{\vec{y}}

\newcommand{\sbold}{\vec{s}}
\newcommand{\lb}{\vec{l}}

\usepackage{algorithmic}

\newcommand{\algorithmicinput}{\textbf{Input:}}
\newcommand{\INPUT}{\item[\algorithmicinput]}

\newcommand{\algorithmicinit}{\textbf{Initialisation:}}
\newcommand{\INITIALIZATION}{\item[\algorithmicinit]}

\begin{document}
\title{Learning Strong Substitutes Demand via Queries}
\author{Paul W. Goldberg}
\email{paul.goldberg@cs.ox.ac.uk}
\orcid{0000-0002-5436-7890}
\author{Edwin Lock}
\email{edwin.lock@cs.ox.ac.uk}
\orcid{0000-0002-0604-2602}
\author{Francisco J. Marmolejo-Cossío}
\email{francisco.marmolejo@cs.ox.ac.uk}
\orcid{0000-0003-3219-7963}
\affiliation{%
    \institution{University of Oxford}
    \department{Department of Computer Science}
    \city{Oxford}
    \country{United Kingdom}
}


\begin{abstract}
This paper addresses the computational challenges of learning strong substitutes demand when given access to a demand (or valuation) oracle. Strong substitutes demand generalises the well-studied gross substitutes demand to a multi-unit setting.  Recent work by Baldwin and Klemperer shows that any such demand can be expressed in a natural way as a finite list of weighted bid vectors. A simplified version of this bidding language has been used by the Bank of England.

Assuming access to a demand oracle, we provide an algorithm that computes the unique list of weighted bid vectors corresponding to a bidder's demand preferences. In the special case where their demand can be expressed using positive bids only, we have an efficient algorithm that learns this list in linear time. We also show super-polynomial lower bounds on the query complexity of computing the list of bids in the general case where bids may be positive and negative. Our algorithms constitute the first systematic approach for bidders to construct a bid list corresponding to non-trivial demand, allowing them to participate in `product-mix' auctions.
\end{abstract}


\maketitle

\section{Introduction}
The Product-Mix Auction \cite{Kle2008, Kle2010, Kle2018} was devised by Klemperer as a means of providing liquidity to commercial banks and has been used regularly by the Bank of England since~2011. In it, there are a number of distinct {goods} available in multiple discrete units, and a set of buyers who express {strong substitutes} demands amongst these goods.%
\footnote{In the banking context, the goods correspond to liquidity secured against alternative kinds of collateral. Commercial banks pay for liquidity `products' by committing to interest rates. The values of the bids submitted by the commercial banks correspond to the interest rates they are willing to pay.}
Given these strong substitutes constraints on the bidders' demand, it is possible to compute market-clearing prices and allocations, in the sense that all buyers receive an allocation that they demand at those prices, and all goods are sold \cite{BGKL19}. The strong substitutes property guarantees the existence of a competitive equilibrium.

Importantly for the present paper, the auction introduces a novel bidding language in which buyers express their demands in terms of lists of bids, where each bid consists of a price vector (one price for each good) and a weight. Any bid $\bid$ is understood as a willingness to buy some quantity of goods (the weight of $\bid$), and for each good $i$ a price $b_i$ is offered. A bid is rejected if all prices offered are lower than the market-clearing prices of the corresponding goods, otherwise it is accepted on some good that maximises the price offered minus the market-clearing price.\footnote{In this way, a single bid with vector $\bid$ and weight $1$ is functionally equivalent to a unit-demand bidder with valuation $\bid$ and quasi-linear utility.} The auction currently run by the Bank of England only permits bidders to submit bids with one non-zero vector entry and positive weights, and it is straightforward to show that any such list of such positive bids has the strong substitutes property. Conversely, it has subsequently been shown that \emph{any} strong-substitutes demand function can be uniquely represented as a list of bids with positive \emph{and} negative weights \cite{Bal-Kle}.

While this gives the buyer a general-purpose means of communicating any strong substitutes demand, the buyer faces the problem of expressing her demand in this language. It may be easier for a buyer to answer queries of the form ``What bundle would you demand, given the following per-unit prices of goods?''. In this paper, we develop query protocols that assist a buyer in constructing her demand function based on a sequence of such queries. Given an unknown demand function, our algorithms are assumed to have access to a demand oracle: for any given prices for goods, the algorithms can learn a bundle of goods demanded at those prices. We are interested in minimising the number of queries to the demand oracle.

\subsection{Our Contributions}
This paper addresses the computational challenges of learning strong substitutes demand when given access to a demand oracle. Under the mild assumption that bidders are able to answer questions of the form ``What bundle do you demand at the following per-unit prices?'', our algorithms constitute the first systematic approach for bidders to generate a bid list corresponding to their demand, allowing them to participate in Product-Mix Auctions with non-trivial demand preferences. We provide upper and lower bounds on the query complexity of learning demand preferences and expressing these using the bidding language of the Product-Mix Auction, which is able to encode any strong substitutes demand in a conceptually simple and natural fashion.

Section \ref{sec:preliminaries} outlines three complementary characterisations of the strong substitutes property and introduces the bidding language both algebraically and geometrically. A first result of this paper, given in Section~\ref{sec:simulating}, is to show that demand oracles are not unreasonably powerful: when given access instead to a valuation oracle, it is possible to simulate a demand oracle with $O(n^3 \log(W/n))$ valuation queries, where $n$ is the number of goods and $W$ is the maximum weight of a bid vector.

In Sections~\ref{sec:positivebids} and~\ref{sec:generalbids}, we consider algorithms that learn the unique bid list corresponding to a bidder's demand. The algorithm in Section~\ref{sec:positivebids} learns demands that can be represented by lists of positive bids, and has linear query complexity. In the setting where demand may require positive and negative bids to express, we provide an exponential-cost algorithm that proceeds by learning all hyperplanes that contain facets of the Locus of Indifference Prices (LIP), a geometric object introduced by Baldwin and Klemperer \cite{BK} to characterise~demand.

Finally, in Section~\ref{sec:lowerbounds} we consider lower bounds on the query complexity of learning bid lists. We note briefly that $\Omega(B \log M)$ queries are required to learn a list of $B$ positive bids, where $M$ is the magnitude of the bid vectors w.r.t.~the $L_\infty$ norm. In order to identify the dependence on the number of goods $n$, we construct an adversarial game using a novel `island gadget' consisting of bids with weight $\pm 1$. Crucially, the island gadget only changes demand in a local region. For fixed $n$, we identify the overall query complexity of learning bid lists corresponding to strong substitutes demand as $\Theta(B \log M + B^n)$.

\subsection{Related Work}\label{sec:related}
Our work relates to the theory of preference elicitation. In this setting, a centralised agent, such as an auctioneer, wishes to identify an optimal allocation of goods via queries to participants' preferences. Queries typically take the form of value queries, where an agent reports a valuation for a given bundle of goods, or demand queries, where an agent reports a bundle that is demanded at given prices. This paper focuses on using demand queries to learn the bid  list representation of strong substitutes demand preferences. This representation can then be used to compute an optimal allocation of goods to agents via the methods of~\cite{BGKL19}.


Much early work in preference elicitation highlights the deep connections to exact learning via membership and equivalence queries from computational learning theory, and our results can also be viewed through this lens. Some notable examples include \cite{zinkevich2003polynomial, BJSZ04}. The authors of \cite{ConenSandholm01} explore the use of ranking oracles to exploit the topological structure of bidder preferences to learn optimal allocations. This approach is extended in \cite{ConenSandholm02differential, ConenSandholm02partial}, and verified empirically in \cite{hudson2004effectiveness}. The authors of \cite{NisanSegal06} explore the communication complexity of preference elicitation in combinatorial auctions, where they show that for general valuations, finding a value-maximising allocation requires an exponential communication cost in the number of items. In \cite{LahaieParkes04}, the authors explore connections between preference elicitation and exact learning, but they demonstrate that the representation length of the valuation is an important parameter in the query complexity of computing optimal allocations. This dependence on the representation length provides a way of side-stepping lower bounds from \cite{NisanSegal06}, and further justifies the need for succinct yet expressive bidding languages, as explored further in \cite{nisan2000bidding}.

Our problem is conceptually similar to a problem studied recently in \cite{ZC20}, in which the authors also consider algorithms with access to demand queries, and attempt to learn the underlying valuation that gives rise to the demand correspondence. The main difference between their work and ours however is that they consider a different class of value functions. In \cite{ZC20}, there are $n$ goods, each in unit supply (whereas we allow multiple copies of goods), and the buyer wants at most $k$ goods, and has additive valuations (whereas our strong substitutes valuations are more general).

\section{Preliminaries}
\label{sec:preliminaries}
We denote $[n] \coloneqq \{1, \ldots, n\}$ and $[n]_0 \coloneqq \{0, \ldots, n\}$. In our auction model, there are $n$ distinct goods numbered from $1$ to $n$; a single copy of a good is an \emph{item}. A \emph{bundle} of goods, typically denoted $\xb$ or $\yb$ in this paper, is a vector in $\Zn_+$ whose $i$-th entry denotes the number of items of good $i$. Vectors $\pb, \qb \in \mathbb{R}^n$ typically denote vectors of prices, with a price entry for each of the $n$ goods. We write $\pb \leq \qb$ when the inequality holds component-wise. Occasionally, it is convenient to work with a notional \emph{reject good} 0 for which prices are always zero. The set of goods is then $[n]_0$ and we identify bundles and prices with the $n+1$-dimensional vectors obtained by adding a $0$-th entry of value 0. For any subset $X \subseteq [n]$, $\eb^X$ denotes the characteristic vector of $X$, i.e.~an $n$-dimensional vector whose $i$-th entry is 1 if $i \in X$, and 0 otherwise. Furthermore, $\eb^i$ denotes the vector whose $i$-th entry is~1 and other entries are 0.

For any vector $\vb \in \mathbb{R}^n$, the $L_1$ and $L_\infty$ norms are defined as $\| \vb \|_1 = \sum_{i \in [n]} |v_i|$ and $\| \vb \|_\infty = \max_{i \in [n]} |v_i|$. The $L_\infty$ ball $B^\varepsilon_{\pb}$ of radius $\varepsilon$ at $\pb$ consists of all points $\qb \in \mathbb{R}^n$ that satisfy $\|\pb-\qb\|_\infty \leq \varepsilon$; note that this a hypercube with edge length $2\varepsilon$. Any hypercube centred at $\pb$ can be partitioned into $2^n$ \emph{orthants}, where every orthant $O_a$ is described by some vector $\vec{a} \in \{-1,1 \}^n$ and consists of the set of points $O_a = \pb + \left \{ \xb \in \mathbb{R}^n \mid a_i x_i \geq 0, \forall i \in [n] \right \}$.%
\footnote{Orthants in $n$-dimensional space generalise the notion of quadrants and octants in two- and three-dimensional space, respectively.}
Every such orthant can be triangulated into $n!$ simplices as follows. For every ordering $[i_1, \ldots, i_n]$ of the indices $[n]$, we define a simplex as the set of points in the orthant that satisfy $a_{i_1}x_{i_1} \leq a_{i_2}x_{i_2} \leq \ldots \leq a_{i_n}x_{i_n}$.

\subsection{Strong-Substitutes Demand Preferences}
\label{sec:SS-demand}
Throughout, we assume that bidders have quasi-linear \emph{strong substitutes} (SS) demand. The SS property is appealing because it is a generalisation of \emph{gross substitutes} (GS) from the single-unit setting that guarantees the existence of a competitive equilibrium in multi-unit auction markets.
We first present a characterisation of SS by Shioura and Tamura \cite{ShiouraTamura2015} that elucidates the relationship between GS and SS before introducing the two equivalent characterisations that underpin our algorithmic results and draw from tropical geometry and discrete convex analysis, respectively. For a detailed survey on the relationship between GS and SS, we refer to Shioura and Tamura \cite{ShiouraTamura2015}.

We assume that bidders have an implicit \emph{valuation} $v : A \to \R$ for bundles of goods, where $A \subset \Zn_+$ is a finite set. This is equivalent to defining the valuation as $v : \Zn \to \uR$, where $\uR \coloneqq \mathbb{R} \cup \{-\infty \}$ denotes the partially extended reals, and we assume that the \emph{effective domain} ${\dom v = \{ \xb \in \Zn \mid v(x) > - \infty \}}$ of $v$ is finite and non-negative in the sense that $\xb \geq \bm{0}$ for all $\xb \in \dom v$. Moreover, bidders have quasi-linear utilities, i.e.~the utility they derive from bundle $\xb$ at prices $\pb$ is $u(\xb; \pb) \coloneqq v(\xb) - \pb \cdot \xb$. A bidder's \emph{demand correspondence} $D_v$ maps prices~$\pb$ to the set of bundles~$\xb \in A$ that maximise $u(\xb; \pb)$ for this $\pb$.

Definition~\ref{def:strong-substitutes} captures how GS and SS demand changes when we (weakly) increase prices. Intuitively, the GS property states that the bidder's demand for those goods with unchanged prices does not decrease, while the law of aggregate demand (LAD) guarantees that the overall number of items that are demanded does not increase. The SS property combines the GS property with LAD.

\begin{definition}[cf.~\cite{ShiouraTamura2015}]
\label{def:strong-substitutes}
A demand correspondence $D_v$ is \emph{gross substitutes} (GS) if, for any prices $\pb' \geq \pb$ with $D_v(\pb)=\{\xb\}$ and $D_v(\pb')=\{\xb'\}$, we have $x'_k \geq x_k$ for all $k$ such that~$p_k=p'_k$. $D_v$ is \emph{strong substitutes} (SS) if $\xb$ and $\xb'$ additionally satisfy $\|\xb'\|_1 \leq \|\xb\|_1$ (the law of aggregate demand).
\end{definition}

\subsubsection{Geometric Approach}\label{sec:lip}
We give some geometric intuition for strong substitutes demand correspondences that underpins the algorithmic ideas in this paper. It is well-known that any quasi-linear demand divides price space into piecewise-linear convex regions corresponding to bundles. When demand is SS, each such region is a convex lattice~\cite{Murota2013}. Figure~\ref{fig:polyhedral-complex} illustrates this.

Recently, Baldwin and Klemperer \cite{BK} proposed a new way of characterising demand types. Borrowing from the tropical geometry literature, they introduce the \emph{Locus of Indifference Prices} (LIP), a piecewise-linear geometric object consisting of the set of all prices at which the bidder is indifferent between two or more bundles. They show that the LIP corresponds in a natural way to a polyhedral complex with $n-1$-dimensional facets. In Figure \ref{fig:polyhedral-complex}, the LIP is drawn using dashed lines. Noting that the orientation of the separating facet between two adjacent demand regions characterises how demand changes when moving from one region to the other, Baldwin and Klemperer~\cite{BK} propose a new way of defining demand types by the set of facet-normal vectors of the LIP. In this new paradigm, the strong substitutes demand type is defined as the family of demand correspondences whose LIP facets are normal to $\eb^i$ or $\eb^i - \eb^j$ for some $i,j \in [n]$.
In two dimensions, facets of SS LIPs are either horizontal, vertical or normal to $(1,-1)$. Hence it follows directly from this definition that the demand correspondence in Figure~\ref{fig:polyhedral-complex} enjoys the strong substitutes property.

\begin{figure}
    \centering
    \subfigure[]{\includegraphics[scale=1]{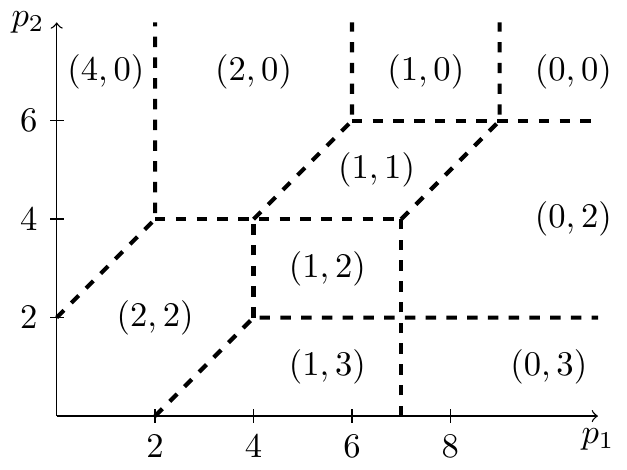}}
    \subfigure[]{\includegraphics[scale=1]{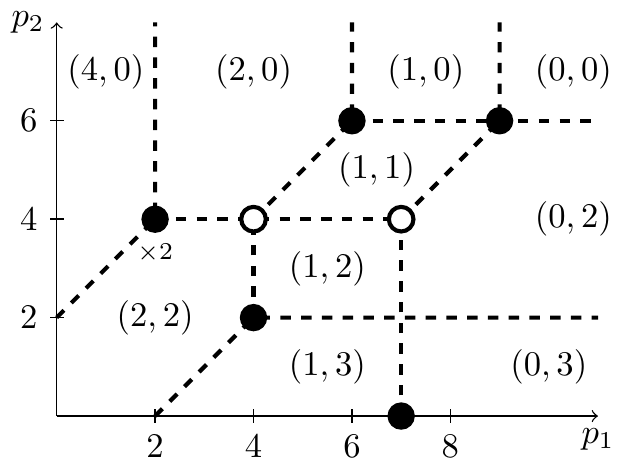}}
    \caption{Left: An illustration of a strong-substitutes demand correspondence with two goods, partitioning price space into piecewise-linear convex regions. Each region is labelled with the bundle demanded at prices in the region. The dashed lines comprise the Locus of Indifference Prices (LIP). Right: Six positive (solid) and two negative (hollow) bids of unit weight are required to express this demand. (Note the two positive unit-weight bids at $(2,4)$, which can alternatively be normalised to a single bid of weight $2$.)}
    \label{fig:polyhedral-complex}
\end{figure}

\subsubsection{Discrete Convex Analysis}
A `price-free' characterisation of the strong substitutes property using the language of discrete convex optimisation is given by Shioura and Tamura \cite{ShiouraTamura2015}. A function $f : \Zn \to \uR$ is called \emph{$M^\natural$-concave} if it satisfies the following \emph{exchange property}. For any $\xb, \yb \in \dom f$ and $i \in \supp^+(\xb - \yb)$, there exists ${j \in \supp^-(\xb - \yb) \cup \{ 0 \}}$ such that
\begin{equation}\label{eq:exchange-property}
    f(\xb) + f(\yb) \leq f(\xb - \eb^i + \eb^j) + f(\yb + \eb^i - \eb^j).
\end{equation}
Here we define $\eb^0 = 0$, and the positive and negative support of a vector $\vec{z} \in \Zn$ as $\supp^+(\vec{z}) = \{ i \in [n] \mid z_i > 0 \}$ and $\supp^-(\vec{z}) = \{ i \in [n] \mid z_i < 0 \}$.

\begin{theorem}[{\cite[Theorem 4.1]{ShiouraTamura2015}}]
\label{thm:SS-is-M-natural-concave}
    A quasi-linear demand correspondence $D_u$ is strong substitutes if and only if its valuation $u$ is $M^\natural$-concave.
\end{theorem}

$M^\natural$-concave functions are closely related to $M$-concave functions, which satisfy the exchange property \eqref{eq:exchange-property} for some non-zero $j \in \supp^-(\xb - \yb)$. Every $n$-dimensional $M^\natural$-concave function can be obtained as the projection of an $n+1$-dimensional $M$-concave function onto an $n$-dimensional hyperplane. Conversely, we can obtain the corresponding $M$-concave function $\hat{f}$ of an $M^\natural$-concave function $f$ as
\begin{equation}\label{eq:M-concave}
    \hat{f}(x_0, \xb) =
    \begin{cases}
    f(\xb)  & \text{if } x_0 = - \sum_{i \in [n]} x_i,\\
    -\infty & \text{otherwise},
    \end{cases}
\end{equation}
where $(x_0, \xb) \in \Z^{n+1}$ is an $n+1$-dimensional vector. For details on $M^\natural$- and $M$-concave functions, we refer to Murota \cite{Murota2013}.

\subsection{The Bidding Language}\label{sec:dot-bids}
The Product-Mix Auction introduces a novel bidding language that allows us to express every strong substitutes demand as a finite list $\bids$ of positive and negative bids. A \emph{bid} consists of an $n$-dimensional integral vector $\bid \in \Z^n$ and a weight $w(\bid) \in \mathbb{Z}$. When working with the notional reject good $0$ introduced above, we identify a bid vector $\bid$ with the $n+1$-dimensional vector obtained by adding a $0$-th entry of value 0. We note that any bid with a weight of $w(\bid) \in \Z$ is equivalent to $w(\bid)$ unit bids with the same vector and sign. This allows us to normalise bid lists to their most succinct form, where no two bids share the same vector. In this paper, we wish to learn the unique normalised bid list that represents the bidder's demand correspondence.


For each bid $\bid$, we can understand $b_i$ as the amount that $\bid$ is willing to spend on good~$i$. Suppose the auctioneer sets prices $\pb$. The bid is \emph{rejected} at~$\pb$ if $b_i < p_i$ for all goods~$i$. Otherwise, the bid \emph{demands a good} $i \in [n]$ that maximises $b_i - p_i$ at price $\pb$. The notational `reject' good $0$ simplifies notation: recalling that we defined $b_0 = 0 = p_0$, we say that~$\bid$ demands good $i \in [n]_0$ if $i \in \argmax_{i \in [n]_0} (b_i - p_i)$, and receiving the `reject' good is equivalent to the bid being rejected. If the set of demanded goods $\argmax_{i \in [n]_0} {(b_i - p_i)}$ at~$\pb$ contains more than one good, we say that $\bid$ is \emph{indifferent} between these goods at~$\pb$. (In particular, a bid may be indifferent between demanding goods and being rejected when ${\max_{i \in [n]_0} (b_i - p_i) = 0}$). A price~$\pb$ is \emph{marginal} if there are bids indifferent between goods at~$\pb$, and non-marginal otherwise.

We can now introduce the \emph{demand correspondence} $D_\bids(\pb)$ for a bid list $\bids$ as follows. If $\pb$ is non-marginal, the unique bundle demanded at $\pb$ is obtained by adding $w(\bid)$ items of $i(\bid)$ to the bundle for each $\bid \in \bids$, where $i(\bid)$ is the unique good that $\bid$ demands at $\pb$. If $\pb$ is marginal, $D_\bids(\pb)$ consists of the discrete convex hull of the bundles demanded at non-marginal prices arbitrarily close to~$\pb$, where the discrete convex hull of a set of bundles $X$ is defined as $\conv(X) \cap \Z$. In general, this implies that we cannot independently allocate to each bid one of the goods it demand, as this may result in bundles that are not in $D_\bids(\pb)$.

Baldwin and Klemperer \cite{Bal-Kle} show that any strong substitutes demand correspondence $D_v$ can be represented as a finite list $\bids$ of positive and negative bids such that $D_v(\pb) = D_\bids (\pb)$ for all prices $\pb$, and this representation is essentially unique (if we restrict ourselves to normalised bid lists as described above). The bids in Figure~\ref{fig:polyhedral-complex} (right) represent the strong substitutes demand shown in Figure~\ref{fig:polyhedral-complex} (left). Conversely, however, not all lists of positive and negative bids induce a strong substitutes demand correspondence; we call a bid list \emph{valid} if it does. Theorem~\ref{thm:valid-bid-list}, taken from \cite{BGKL19}, gives a criterion that allows us to check validity. It is known that the problem of deciding the validity of a bid list is coNP-complete~\cite{BGKL19}.

\begin{theorem}
\label{thm:valid-bid-list}
    A bid list is valid if and only if the weights of the bids indifferent between $i$ and $i'$ at $\pb$ sum to a non-negative number, for all $\pb \in \Rn$ and $i,i' \in [n]_0$.
\end{theorem}

A special subclass of the strong substitutes demand type is the family of demand correspondences that can be expressed using only positive bids. This family is of particular practical interest, as the Bank of England currently runs the Product-Mix Auction with positive bids only. Moreover, as a single positive bid corresponds to the demand of a unit-demand consumer, learning a list of positive bids is equivalent to learning the demands of a collection of unit-demand consumers; the latter have been studied in a variety of settings such as profit-maximising envy-free pricing \cite{Guruswami2005}. 
Note that any list of positive bids is valid, as it trivially satisfies Theorem~\ref{thm:valid-bid-list}.

\subsection{The Geometry of Bids}
In the previous section, we explained the algebraic relationship between a bid list and its resulting demand correspondence. Here we highlight the geometry of such a demand correspondence, as this forms the basis of our algorithms.

Fix a bid $\bid \in \Zn$ and an item $i \in [n]_0$. We let $R_i$ denote the set of prices at which $\vec{b}$ demands good~$i$. Each $R_i$ is an unbounded convex polytope in $\Rn$, which can be expressed succinctly as $R_i = \bid + H_i$, where $H_i$ is the conic hull of the vectors $\eb^j$, $j \in [n]_0 \setminus \{i\}$, if we define $\eb^0 = -\eb^{[n]}$ (see Figure~\ref{fig:exa}). Note that each~$R_i$ is of full affine dimension and the $R_i$ together cover the entirety of $\mathbb{R}^n$. In line with the previous section, $\bid$ is indifferent between two goods $i, j \in [n]_0$ at $\pb$ if and only if $\pb \in R_i \cap R_j$. Moreover, if a price lies in the interior of any given $R_i$, good~$i$ is the unique good demanded by $\bid$. For a given list of bids $\bids$, we recall that a price $\vec{p}$ is non-marginal if each bid in $\bids$ demands a unique item at $\vec{p}$. If we define $R^j_i \coloneqq \bid^j + H_i$ for each $\bid^j \in \bids$, then it is straightforward to see that $\vec{p}$ is non-marginal if and only if it does not lie on the boundary of any $R^j_i$. This allows us to describe the geometry of the demand correspondence $D_\bids$. Suppose that $\pb$ is non-marginal, and that for each $\bid \in \bids$, $i(\bid)$ is the unique good demanded by $\bid$ at $\pb$. We recall that the unique bundle demanded at~$\vec{p}$ is obtained by adding $w(\vec{b})$ items of good $i(\vec{b})$ for each $\vec{b} \in \bids$. We can express this as
\[
D_{\bids}(\vec{p}) = \sum_{\substack{(i,j): \\ \vec{p} \in R^j_i}} w(\vec{b}^j) \vec{e}^i.
\]
This is illustrated in Figure~\ref{fig:exa}. We extend the definition of~$D_\bids$ to marginal prices as above: if $\vec{p}$ is marginal, $D_{\bids}(\vec{p})$ is the discrete convex hull of the bundles demanded at non-marginal prices arbitrarily close to $\vec{p}$.

\begin{figure}[tb!]
    \centering
    \subfigure[]{
        \includegraphics[scale=1]{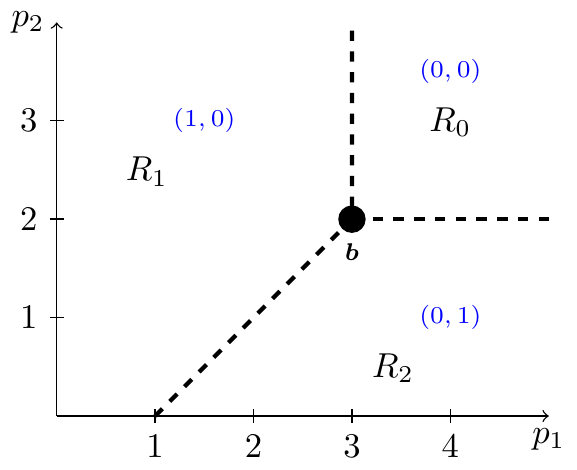}
    }
    \subfigure[]{
        \includegraphics[scale=1]{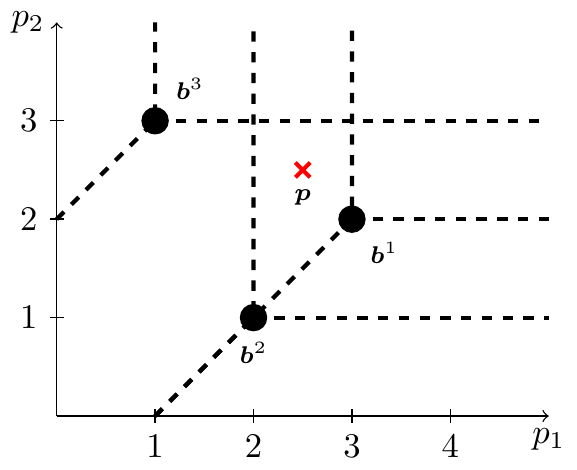}
    }
    \caption{Left: The demand regions $R_i$ of a single positive bid $\vec{b}$ in the case of two goods $1,2$ and the reject good $0$.
    Right: a three-element bid list $\bids = \{\bid^1,\bid^2,\bid^3 \}$.
    The price vector $\vec{p}$ (marked by a red cross) lies in $R_1^1$, $R_0^2$, and $R_2^3$.
    Hence, at prices $\pb$, the bid $\vec{b}^1$ demands good $1$, $\vec{b}^2$ is rejected, and $\vec{b}^3$ demands good 2.
    Putting this together, $D_{\bids}(\vec{p}) = (1,1)$. }
    \label{fig:exa}
\end{figure}

\subsubsection{Valid Bid Lists}
As mentioned above, not all demand correspondences $D_\bids$ arising from bid lists encode a quasi-linear, strong substitutes demand correspondence $D_v$ as defined in Section~\ref{sec:SS-demand}. To obtain SS demand, the bid list must satisfy the property given in Theorem \ref{thm:valid-bid-list}. Here we explore what validity means geometrically. Suppose that $\bids$ is a bid list, and for any given price $\vec{p}$, define the $(i,\ell)$-\emph{support} $\supp_{i,\ell} (\vec{p})$ of $\vec{p}$ to be the set of all bids $\vec{b}^j \in \bids$ such that $\vec{p} \in R^j_i \cap R^j_\ell$.
In other words, all bids that are indifferent between goods $i$ and $\ell$ at price $\vec{p}$.
This allows us to give a condition for non-validity of our bid list $\bids$. In particular, $\bids$ is valid if and only if $\sum_{\vec{b} \in \supp_{i,\ell} (\vec{p})} w(\vec{b}) \geq 0$ for all prices $\pb$ and pairs of goods $i, \ell \in [n]_0$.

\subsection{The Computational Challenges}
Consider a bidder who has an (unknown) strong substitutes demand correspondence $D_v$ on~$n$ goods. We study the problem of learning the unique list $\bids$ of positive and negative bids of unit weight that represent a bidder's demand correspondence, i.e.~such that $D_v = D_\bids$. We consider algorithms that learn $\bids$ by querying the demand correspondence $D_v$ at different price vectors. More specifically, our algorithms have access to an \emph{adversarial demand oracle}~$\mathcal{Q}_\bids$; given any price vector $\pb$, $\mathcal{Q}_\bids(\pb)$ returns a bundle from $D_\bids(\pb)$. A bidder may demand multiple bundles at some price (i.e.\ when $|D_\bids(\pb)| > 1$), in which case the adversarial oracle simply returns a single demanded bundle at that price, and we have no control over which such bundle is returned. Another related setting we address in Section~\ref{sec:simulating} is the complexity of learning $\bids$ given access to a \emph{valuation oracle}, i.e.~given a bundle $\xb$, the bidder reports their valuation $v(\xb)$ for this bundle.

Let $B \coloneqq |\bids|$ be the number of bids we wish to learn. Moreover, let $M \coloneqq \max_{\bid \in \bids} \|\bid \|_\infty$ be the \emph{magnitude} of the bids w.r.t to the $L_\infty$ norm and $W \coloneqq \max_{\bid \in \bids} w(\bid)$ be the maximum bid weight. Note that for any unknown bid list $\bids$, we can determine the value of $M$ with $O(\log M)$ demand queries, as $M$ corresponds to the smallest value $m$ such that the bidder demands the empty bundle at price vector $\pb = m \eb^{[n]}$, which can be found using binary search. We are interested in the query complexity of learning $\bids$, measured in terms of $n$, $B$, $\log M$ and $\log W$. Note that $n B\log M + B \log W$ bits are required to store the bid list $\bids$, under the natural assumption that bid vectors and weights are encoded in binary.

\section{Simulating \texorpdfstring{$\qoracle$}{a Demand Oracle} with a Valuation Oracle}
\label{sec:simulating}

In this section we show that demand oracles are not unreasonably powerful, in the sense that we can use a valuation oracle to simulate a demand oracle with polynomial overhead. Consider the setting where we are given query access to a bidder's valuation function $v$. We show that a single query to $\qoracle$ can be simulated with a polynomial number of queries to a valuation oracle. This result utilises the equivalence of the strong substitutes property and $M^\natural$-convexity from the discrete convex analysis literature.

Recall that the utility of bundle $\xb$ at prices $\pb$ is given by $u(\xb;\pb) = v(\xb) - \pb \cdot \xb$. We define $u_{\pb} \coloneqq u(\cdot;\pb)$ for convenience. In order to simulate a demand oracle on input $\pb$, we wish to compute a bundle $\xb \in D_v$ that maximises $u_{\pb}(\xb)$. Note that we can compute $u_{\pb}(\xb)$ for any bundle $\xb$ using a single query to the valuation oracle. In order to compute a maximiser of $u_{\pb}(\cdot)$, we draw from the discrete convex analysis literature. Firstly, we see that $u_{\pb}$ is $M^\natural$-concave. Indeed, it is well-known that strong substitutes valuations are $M^\natural$-concave \cite{ShiouraTamura2015} and subtracting a linear term preserves this property. Secondly, let $\hat{u}$ be the corresponding $M$-concave function to $u_{\pb}$ as defined in \eqref{eq:M-concave}. We see that maximising $u_{\pb}$ is equivalent to maximising~$\hat{u}$. Moreover, we can compute $\hat{u}(x_0, \xb)$ using at most one query to the valuation oracle. Thirdly, note that we have $\|x\|_1 \leq BW$ for any bundle $\xb$ that the bidder demands, as every bid $\bid \in \bids$ contributes at most $W$ items to $\xb$.

Murota \cite[Chapter 10]{Murota2013} provides multiple algorithms for maximising $M$-concave functions~$f$ with bounded effective domains $\dom f$. The simplest such algorithm, a straightforward steepest descent method, finds a maximiser with $O(n^2 L)$ queries, where $L \coloneqq \max \{ \|x-y \|_1 \mid x,y \in \dom f \}$. In our setting, we have $L = BW$, yielding a query complexity of $O(n^2BW)$. This query complexity can be improved to $O(n^3 \log(BW/n))$ by applying the more involved algorithms for maximising $M$-concave functions given in \cite{Shioura2004} and \cite{Tamura2004}. We note that this query complexity is polynomial in $n, B$ and $\log W$.

\section{Learning Positive-Weighted Bids}\label{sec:positivebids}
In this section we assume that the bidder's demand correspondence can be expressed by a list of positive bids. This is the case, for instance, when we wish to learn the individual demands of a collection of unit-demand consumers.\footnote{In this setting we would have access only to an oracle that returns an aggregately demanded bundle at given prices.} Our algorithm learns a list of $B$ positive bids using $O(n B \log M)$ demand queries. This is close to our lower bound of $\Omega(B \log M)$ given in Theorem~\ref{thm:simple-lower-bound} below.

We proceed by repeatedly finding a bid and `removing' it, thereby reducing the size of the remaining demand correspondence until all bids have been found. Let $\learnt$ denote the subset of bids from $\bids$ that have already been learnt, and let $\bids' \coloneqq \bids - \learnt$ be the list of remaining bids. We can simulate a demand oracle $\mathcal{Q}_{\bids'}(\pb)$ for the demand correspondence associated with $\bids'$ as follows. At price vector $\pb$, first determine a bundle $\xb$ demanded by all bids in~$\bids$ with a single query $\qoracle(\pb)$, and then subtract from $\xb$ a~bundle~$\yb$ demanded at $\pb$ by the bids in $\learnt$.\footnote{Note that $\bids'$ is valid, as lists of positive bids are always valid. If $\bids$ consisted of positive and negative bids, removing a single positive bid might result in a bid list that is no longer valid. In this case, the algorithm described in this section may fail and return points not corresponding to bid locations.}

In this way, the problem of learning a list of positive bids reduces to repeatedly identifying a single bid. In the next section, we describe a subroutine that learns the location of a single bid in $\bids'$ using $O(n \log M)$ queries. As this subroutine is called~$B$ times, this yields an overall query complexity of $O(nB \log M)$ for learning all bids in $\bids$. Recall that we can compute $M$ with $O(\log M)$ queries.

\subsection{Finding a Single Positive Bid}
\label{sec:finding-positive-bid}
We present an algorithm that performs binary searches using \emph{delta queries} to successively learn the coordinates $x_1, \ldots, x_n$ of a bid's location $\xb$ together with its weight. We begin by defining delta queries and establishing some fundamental facts about the results returned by these queries.

\begin{definition}
\label{def:delta-query}
    A \emph{delta query} $\Delta(\qb)$ at $\qb \in \Zn$ consists of two queries $\qb^+$ and $\qb^-$ defined by $\qb^+ = \qb' + \frac{1}{2n}\eb^1$ and $\qb^- = \qb' - \frac{1}{2n}\eb^1$, where we let
    \[
		\qb' \coloneqq \qb + \sum_{i \in \{ 2, \ldots, n \} } \frac{1}{2(n-i+1)} \eb^i.
	\]
    The return value of the delta query is defined as $\Delta(\qb) = x^-_1 - x^+_1$,
    where $\xb^+$ and $\xb^-$ are the bundles of goods uniquely demanded at $\qb^+$ and $\qb^-$.
\end{definition}

Note that $q^+_1 = q_1 + \frac{1}{2n}$ and $q^-_1 = q_1 - \frac{1}{2n}$, and the two query points $\qb^+$ and $\qb^-$ agree on all other coordinates $i \geq 2$. Secondly, $\qb^\pm$ is non-marginal by construction, so any bid $\bid \in \Zn$ uniquely demands some good $i$ at $\qb^\pm$. The intuition behind delta queries is as follows. Consider the hyperplane normal to $\eb^1$ that contains $\qb$. In a first step, we carefully perturb $\qb$ such that the resulting point $\qb'$ remains on the hyperplane and no bid is indifferent between any two goods in $\{2, \ldots, n\}$. The points $\qb^-$ and $\qb^+$ are then obtained by perturbing $\qb'$ in directions $\pm \eb^1$ such that the prices become non-marginal.

Lemma~\ref{lemma:delta-query} makes the observation that bids $\bid$ satisfying $b_1 = q_1$ and $\bid \leq \qb$ demand good~$1$ at $\qb^-$ and are rejected at $\qb^+$, while all other bids demand the same good at both prices $\qb^\pm$. Hence demand changes only in terms of good $1$, and $\Delta(\qb)$ captures the magnitude of this change.\footnote{Lemma~\ref{lemma:delta-query} still holds if the bid list contains positive \textit{and} negative bids. Indeed, a generalised version of delta queries is used in our algorithm for learning positive and negative bids in Section~\ref{sec:generalbids}.}. In our current setting where all bids have positive weights, Corollary~\ref{corollary:delta-counting} notes that this is equivalent to summing the weights of the bids $\bid$ that satisfy $b_1 = q_1$ and $\bid \leq \qb$. Our algorithm exploits this fact in order to learn the coordinates of a bid location as well as the bid weight.

\begin{lemma}\label{lemma:delta-query}
    Suppose we place a delta query at $\qb \in \Zn$. Then any bid $\bid \in \Zn$ demands different goods at $\qb^-$ and $\qb^+$ if and only if $b_1 = q_1$ and $\bid \leq \qb'$. Moreover, any such bid demands an item of good 1 at $\qb^-$ and an item of the reject good 0 at $\qb^+$.
\end{lemma}
\begin{proof}
    Suppose $\bid$ is a bid with $b_1 = q_1$ and $\bid \leq \qb'$. As $\bid$ is integral, we have ${b_i < q^+_i = q^-_i}$ for all goods $i \geq 2$, which implies that $\bid$ can only demand goods $0$ or $1$ at $\qb^-$ and $\qb^+$. At~$\qb^-$, bid $\bid$ uniquely demands good 1, as we have $b_1 - q^-_1 = \frac{1}{2n} > 0$. Similarly, we see that $\bid$ uniquely demands the reject good~$0$ at $\qb^+$, as $b_1 - q^+_1 = -\frac{1}{2n} < 0$.

    Conversely, suppose $\bid$ demands distinct goods $i$ and $j$ at $\qb^-$ and $\qb^+$, respectively. This implies the two fundamental inequalities (a) $b_i - q_i^- > b_j - q_j^-$ and (b) $b_j - q_j^+ > b_i - q_i^+$. We first show that we must have $i=1$ and $j=0$ by excluding all other possibilities. (Recall that prices $p_0$ and bid values $b_0$ are 0 for the reject good 0, by definition.) Suppose that $i=0$ or $i \geq 2$. Then $q_i^- = q_i^+$ by construction of $\qb$, so (a) and (b) imply $b_j - q_j^- < b_j - q_j^+$ in contradiction to $q_j^- \leq q_j^+$, which holds by construction of $\qb$.
    Next suppose $i=1$ and $j \geq 2$. Then $q_j^- = q_j^+ = q_j - \frac{1}{2(n-j+1)}$. The inequalities (a) and (b) imply $b_1 - q_1 - \frac{1}{2n} < b_j - q_j^- < b_1 - q_1 + \frac{1}{2n}$. As $\bid$ and $\qb$ are integral, it follows that $b_j - q_j^-$ must lie within $\frac{1}{2n}$ of an integral point, in contradiction to $q_j^- = q_j - \frac{1}{2(n-j+1)}$.

    Finally, we show that a bid $\bid$ that demands good $i=1$ at $q^-$ and $j=0$ at $q^+$ satisfies $b_1 = q_1$ and $\bid \leq \qb'$, again using the integrality of $\bid$ and $\qb$. Firstly, (a) and (b) imply $q_1 - \frac{1}{2n} < b_1 < q_1 + \frac{1}{2n}$ and it follows that $b_1 = q_1$. Secondly, we have $b_k < q_k^+$ for all $k \geq 1$, as $\bid$ uniquely demands the reject good $0$ at $\qb^+$, which implies $\bid \leq \qb$ by integrality.
\end{proof}

\begin{corollary}
\label{corollary:delta-counting}
    $\Delta(\qb)$ is the sum of the weights of all bids $\bid \in \bids$ satisfying $b_1 = q_1$ and $\bid \leq \qb$.
\end{corollary}

\subsubsection{The Algorithm}
\begin{algorithm}[tb!]
\caption{Learning Positive Bids}
\label{alg:bid-finder}
\begin{algorithmic}[1]
    \STATE Perform binary search to find largest price $\pb \in \left \{(k, M, \ldots, M) \mid k \in [M]_0 \right \}$ at which $\Delta(\pb) > 0$, and fix $x_1 = p_1$.
    \FOR{$i = 2 \ldots n$}
        \STATE Binary search to find smallest price $\pb \in \left \{ (x_1, \ldots, x_{i-1}, k, M, \ldots, M) \mid k \in [M]_0 \right \}$ at which $\Delta(\pb) > 0$, and fix $x_i = p_i$.
    \ENDFOR
    \RETURN bid vector $\xb = (x_1, \ldots, x_n)$ and weight $\Delta(\xb)$.
\end{algorithmic}
\end{algorithm}

Algorithm~\ref{alg:bid-finder} learns the vector $\xb$ and weight of a single positive bid with $O(n \log M)$ queries. It determines the value of $x_i$, $i \in [n]$, by performing a binary search on line segment $L_i \coloneqq \{(x_1, \ldots, x_{i-1}, z, M, \ldots, M) \mid 0 \leq z \leq M \}$, where the values of $x_1, \ldots, x_{i-1}$ have already been determined and are fixed. As $L_i$ is well-ordered, we can define the `smallest' and `largest' points on~$L_i$ as $\sbold^i \coloneqq (x_1, \ldots, x_{i-1}, 0, M, \ldots, M)$ and $\lb^i \coloneqq (x_1, \ldots, x_{i-1}, M \ldots, M)$.

In a first step, Algorithm~\ref{alg:bid-finder} performs binary search on $L_1$ in order to find the largest point at which demand for good~1 is positive. Note that at any $\pb \in L_1$, no bid demands items of good $i \geq 2$, i.e.~every bid demands an item of good 0 or 1 (or is indifferent between the two). Moreover, the function mapping prices $\pb$ on $L_1$ to the demand of good $1$ at~$\pb$ is monotonically decreasing and changes only at integral points, as the bids are integral. As $B$ items of good 1 are demanded at $\sbold^1$, there is a largest price $\pb^* \in L_1$ at which demand is positive. Hence, we can find $\pb^*$ using binary search on $L_1$ by querying demand at $O(\log M)$ prices of the form $(k, M, \ldots, M)$ with $k \in [M]_0$.

The second kind of binary search uses delta queries to find the smallest point $\pb^*$ for each line segment $L_i$, $i \geq 2$, at which $\Delta(\pb^*)$ is positive. Suppose $i \geq 2$. Corollary~\ref{corollary:delta-counting} implies that $\Delta(\qb)$ restricted to the line $L_i$ is monotonically increasing. We show in the proof of Theorem~\ref{thm:bid-finding-algorithm} that the invariant $\Delta(\lb^i) > 0$ holds when we perform binary search on~$L_i$. Moreover, $\Delta$ only changes in value at integral points along $L_i$, so we can perform binary search to find~$\pb^*$ with $O(\log M)$ delta queries at prices $(x_1, \ldots, x_{i-1}, k, M, \ldots, M)$, where $k \in [M]_0$.

Theorem~\ref{thm:bid-finding-algorithm} establishes the correctness and running time of Algorithm~\ref{alg:bid-finder}. The proof proceeds by induction and makes use of Observation~\ref{observation:bid-restrictions}.

\begin{observation}\label{observation:bid-restrictions}
    Let $i \geq 2$. Suppose Algorithm~\ref{alg:bid-finder} has successfully determined the first ${i-1}$ coordinates $x_1, \ldots, x_{i-1}$ and binary search (using delta queries) finds the smallest point $\pb = (x_1, \ldots, x_i, M, \ldots, M)$ on $L_i$ at which $\Delta(\pb) > 0$. Corollary~\ref{corollary:delta-counting} implies the following.
    \begin{description}
        \item[(a)] None of the bids $\bid \in \bids$ satisfy $\bid \leq \pb$ and $b_i < p_i$.
        \item[(b)] There is at least one bid $\bid$ with $\bid \leq \pb$ and $b_i = p_i$.
    \end{description}
\end{observation}

\begin{theorem}\label{thm:bid-finding-algorithm}
    Algorithm~\ref{alg:bid-finder} returns the vector and weight of a bid using $O(n \log M)$ queries.
\end{theorem}
\begin{proof}
First we show that the algorithm is well-defined. Clearly, the first coordinate $x_1$ can be found, as outlined above. To show that the binary search on $L_i$ using delta queries to find~$x_i$ is well-defined for all $i \geq 2$, it suffices to show that $\Delta(\lb^i)$ is positive. Fix $i \leq 2$ and suppose the algorithm has found $x_1, \ldots, x_{i-1}$. If $\pb$ is the point found by binary searching on line~$L_{i-1}$, then Observation~\ref{observation:bid-restrictions}~(b) tells us that there is at least one bid $\bid$ satisfying $b_1 = p_1$ and $\bid \leq \pb$. Moreover, we have $p_1 = \lb^i_1$ and $\pb \leq \lb^i$ by construction. By transitivity and Corollary~\ref{corollary:delta-counting}, it follows that~$\Delta(\lb^i) > 0$.

Next, we prove that $\xb = (x_1, \ldots, x_n)$ returned by the algorithm corresponds to the vector of a bid and that $\Delta(\xb)$ is its weight. By Observation~\ref{observation:bid-restrictions}~(b) and construction of $x_n$, we know that there is a bid $\bid^*$ that satisfies $\bid^* \leq \xb$ and $b^*_n = x_n$. On the other hand, for any $1 \leq i \leq n-1$, the construction of~$x_i$ together with Observation~\ref{observation:bid-restrictions} (a) implies that no bid $\bid$ satisfies $\bid \leq \xb$ and $b_i < x_i$. As a result, we get $b^*_i = x_i$ for all $1 \leq i \leq n$ and it follows that~$\bid^*$ lies at point $\xb$. Moreover, we see that $\bid^*$ is the only bid in $\bids$ that satisfies $b_1 = x_1$ and $ \bid \leq \xb$. By Corollary~\ref{corollary:delta-counting}, $\Delta(\xb)$ is the weight of $\bid^*$.

Finally, to see that the algorithm has query complexity $O(n \log M)$, note that it performs~$n$ binary searches along lines $L_i$, each of which incurring $O(\log M)$ queries.
\end{proof}

\section{Learning Positive and Negative Bids}\label{sec:generalbids}
In this section, we provide an algorithm for learning valid bid lists that may contain both positive and negative bids. Introducing negative bids significantly complicates learning the location of bids, as negative bids are able to `cancel out' facets. Consider, for instance, a bid $\bid$ of weight 2 at position $(2,4)$ in the setting with two goods. Without negative bids, its horizontal and vertical facets extend indefinitely from $\bid$ in direction $\eb^1$ and $\eb^2$. In contrast, Figure~\ref{fig:polyhedral-complex} demonstrates how two negative bids of weight 1 cancel out $\bid$'s horizontal facet from $(7,4)$ onwards. Similarly, in the same figure the vertical facet of the bid at $(4,2)$ is cancelled out by the negative bid at $(4,4)$. As a result, some bids may only be detectable if we query in their local neighbourhood; for more details we refer to Section~\ref{sec:pos-neg-lower-bound}, where we exploit this phenomenon to show super-polynomial lower bound on learning positive and negative bids.

Recall from Section~\ref{sec:lip} that the demand correspondence of a SS demand (and hence of a valid bid list) corresponds to a polyhedral complex over price space, the LIP, where the boundaries between unique demand regions are $(n-1)$-dimensional facets.
Our algorithm learns the collection of all hyperplanes that contain these facets, as well as each vertex arising from the intersection of $n$ such hyperplanes. We note that every bid must lie at a vertex but, conversely, not every such vertex contains a bid. We introduce {\em bid existence queries} at integral price vectors, which use a set of $2^{n}$ demand queries in the local neighbourhood of the vector in order to check for existence of a bid there. In addition, we define (more powerful and costly) {\em super queries} which for integral points~$\pb$, use a collection of demand queries to provide complete information about the demand correspondence in the local neighbourhood of $\pb$. This allows us to simulate a bid existence query with a super query and also to perform a principled search for new hyperplanes: at each iteration of the algorithm, either the local information around two vertices points us in the direction of a new hyperplane, or we have succeeded in learning all hyperplanes, and thus all bids.

\subsection{Bid Existence Queries}
\label{sec:existence-queries}
Our algorithm computes a set of integral candidate price vectors at which bids may lie. In order to check for any such candidate vector $\pb \in \Zn$ whether there is a bid at this location and, if so, what its weight is, we first introduce a generalisation of the delta queries that were first introduced in Section~\ref{sec:finding-positive-bid}. This in turn allows us to define the \textit{bid existence query}, which consists of $2^{n-1}$ generalised delta queries. In total, we see that as each generalised delta query consists of two elementary demand queries, a total of $2^n$ queries to the demand oracle is sufficient to determine the existence and weight of a bid at a given integral price vector.

\subsubsection{Generalised delta queries}
Recall that a delta query at $\qb \in \Z^n$ works by first perturbing the entries $i \geq 2$ of $\qb$ in a specific way. We now generalise the directions in which these perturbations happen. For any point $\qb \in \Z^n$ and every $S \subseteq \{2, \ldots, n \}$, we define $\qb(S)$ by $q(S)_1 := q_1$ and $q(S)_i := q_i - \frac{1}{2(n-i+1)}$ if $i \in S$ and $q(S)_i := q_i + \frac{1}{2(n-i+1)}$ otherwise for all $i \geq 2$. This allows us to define a \emph{generalised delta query} at~$\qb$ w.r.t.~$S$ as follows.

\begin{definition}
    Let $\qb \in \Zn$ and $S \subseteq \{2, \ldots, n \}$. A (generalised) delta query consists of two queries $\qb^+(S) := \qb(S) + \frac{1}{2n}\eb^1$ and $\qb^-(S) := \qb(S) - \frac{1}{2n}\eb^1$, where we let
    \[
    \qb(S) := \qb - \sum_{i \in S} \frac{1}{2(n-i+1)} \eb^i + \sum_{i \not \in S \cup \{ 1 \}} \frac{1}{2(n-i+1)} \eb^i.
    \]
    The return value of the delta query is defined as $\Delta(\qb;S) = x^-_1 - x^+_1$, where $\xb^+$ and $\xb^-$ are the bundles of goods uniquely demanded at $\qb^+$ and $\qb^-$.
\end{definition}

Note that $q^+(S)_1 = q(S)_1 + \frac{1}{2n}$ and $q^-(S)_1 = q(S)_1 - \frac{1}{2n}$, and the two query points $\qb^+(S)$ and $\qb^-(S)$ agree on all other coordinates $i \geq 2$. Secondly, $\qb^\pm$ are non-marginal prices by construction, so every bid uniquely demands some good $i$.


Lemma~\ref{lemma:generalised-delta-query} makes the observation that bids satisfying $b_1 = q_1$ and $\bid \leq \qb(S)$ demand good $1$ at $\qb^-$ and are rejected at~$\qb^+$, while all other bids demand the same good at both prices $\qb^\pm$. Hence demand changes only in terms of good $1$, and $\Delta(\qb;S)$ captures this change. Another useful interpretation of $\Delta(\qb;S)$ is that it returns the sum of the weights of all bids that satisfy $b_1 = q_1$ and $\bid \leq \qb$.

\begin{lemma}\label{lemma:generalised-delta-query}
    Suppose we make a delta query at $\qb(S) \in \Zn$ w.r.t.~$S \subseteq \{2, \ldots, n \}$. Then any bid $\bid \in \Zn$ demands different goods at $\qb^-$ and $\qb^+$ if and only if $b_1 = q_1$ and $\bid \leq \qb(S)$. Moreover, any such bid demands an item of good 1 at $\qb^-$ and an item of reject good 0 at $\qb^+$.
\end{lemma}
\begin{proof}
Analogous to the proof of Lemma~\ref{lemma:delta-query}.
\end{proof}

\subsubsection{Defining Bid Existence Queries}
In order to check for the existence (and weight) of a bid at a given price vector $\pb \in \Zn$, we introduce the existence query consisting of $2^{n-1}$ generalised delta queries. Indeed, an existence query performed at $\pb \in \Zn$ consists of the set of delta queries $\{ \Delta(\pb;S) \mid S \subseteq \{2, \ldots, n\}$. We now show how the existence and weight of a bid at $\pb$ can be inferred from the return values of these $2^{n-1}$ delta queries. For convenience, we introduce the notation $w(\mathcal{C}) = \sum_{\bid \in \mathcal{C}} w(\bid)$ for any list of bids $\mathcal{C}$. Our approach is based on the following straightforward observation.
\begin{observation}
\label{obs:existence-condition}
    For any bid list $\bids$, define the two sub-lists $\bids' = \{ \bid \in \bids \mid b_1 = p_1, \bid \leq \pb \}$ and $\bids'' = \{ \bid \in \bids \mid b_1 = p_1, \bid \leq \pb, \bid \neq \pb \}$. There exists a bid at $\pb$ in $\bids$ if and only if the term $w := w(\bids') - w(\bids'')$ is non-zero. Moreover, if the bid exists, then its weight is $w$.
\end{observation}

Note that $\bids'$ contains all bids dominated by $\pb$ and we can compute $w(\bids')$ with a single delta query $\Delta(\pb, \emptyset)$. Moreover, the set $\bids''$ contains all bids that are dominated by $\pb$ and do not lie at $\pb$. In order to determine $w(\bids'')$, we first partition $\bids''$ into $2^{n-1}-1$ sub-lists as follows. For each non-empty set of goods $T \subseteq \{2, \ldots, n\}$, we let $\bids_T := \{ \bid \in \bids \mid b_i < p_i, \forall i \in T \text{ and } b_j = p_j, \forall j \not \in T \}$. Lemma~\ref{lemma:computing-existence-query} now provides us with a method of computing $w(\bids'')$ using the $2^{n-1}-1$ delta queries $\Delta(\pb,S)$ with $\emptyset \subsetneq S \subseteq \{2, \ldots, n\}$.

\begin{lemma}\label{lemma:computing-existence-query}
    We have
    \begin{align*}
        w(\bids'') = \sum_{\emptyset \subsetneq T \subseteq \{2, \ldots, n\}} w(\bids_T) = \sum_{\emptyset \subsetneq S \subseteq \{2, \ldots, n\}} (-1)^{|S|+1} \Delta(\pb;S).
    \end{align*}
\end{lemma}
\begin{proof}
    Firstly, observe how the result of the delta query $\Delta(\pb;S)$ for a given $S \subseteq \{2, \ldots, n\}$ satisfies $\Delta(\pb;S) = \sum_{T \supseteq S} w(\bids_T).$
    Indeed, recall that $\Delta(S)$ sums the weights of the bids $\bid \in \bids$ satisfying $b_1 = p_1$ and $\bid \leq \pb(S)$. As we have $p(S)_i < p_i$ for all $i \in S$ and the bids are integral, the set of bids that $\Delta(\pb;S)$ considers is identical to $\bigcup_{T \supseteq S} \bids_T$. Hence
    \begin{align*}
        \sum_{\emptyset \subsetneq S \subseteq \{2, \ldots, n\}} (-1)^{|S|+1} \Delta(\pb;S)
        &= \sum_{\emptyset \subsetneq S \subseteq \{2, \ldots, n\}} (-1)^{|S|+1}\sum_{T \supseteq S} w(\bids_T) \\
        &= \sum_{\emptyset \subsetneq T \subseteq \{2, \ldots, n\}} w(\bids_T)\sum_{S \subseteq T} (-1)^{|S|+1} \\
        &= \sum_{\emptyset \subsetneq T \subseteq \{2, \ldots, n\}} w(\bids_T).
    \end{align*}
    The last equality follows from the binomial identity in Lemma~\ref{lemma:binom-identity} below.
\end{proof}

\begin{lemma}
\label{lemma:binom-identity}
For any $n \geq 1$ we have $\sum_{i=1}^n (-1)^{i+1} \binom{n}{i} = 1$.
\end{lemma}
\begin{proof}
    Recalling Pascal's Rule, $\binom{n}{i} = \binom{n-1}{i} + \binom{n-1}{i-1}$, we have
    \begin{align*}
        \sum_{i=1}^n (-1)^{i+1} \binom{n}{i} &= \sum_{i=1}^n (-1)^{i+1} \binom{n-1}{i} + \sum_{i=1}^n (-1)^{i+1} \binom{n-1}{i-1} \\
        & = \sum_{i=1}^{n-1} (-1)^{i+1} \binom{n-1}{i} + \sum_{i=1}^{n-1} (-1)^{i} \binom{n-1}{i} + (-1)^2 \binom{n-1}{0} = 1.
    \end{align*}
\end{proof}

\subsection{Super Queries}
\label{sec:superquery}
Suppose $\vec{p} \in \Zn$ is an integral price vector. We show that it is possible to obtain complete knowledge of~$D_\mathcal{B}(\vec{p}')$ for all prices $\vec{p}'$ with ${\|\vec{p} - \vec{p}'\|_\infty < 1}$ using a \emph{super query}, which consists of a specific set of demand queries at non-marginal query points in the vicinity of $\pb$. Intuitively, this works because bid vectors are integral and facets of an SS LIP can only have specific orientations. Super queries are used by our algorithm in two ways: firstly to determine the existence and weight of a bid at a given integral point $\pb$, and secondly to provide information that leads to a new separating hyperplane.

Let $U_1(\vec{p}),\ldots,U_{2^n}(\vec{p})$ denote the $2^n$ orthants of the unit $L_\infty$-ball around $\vec{p}$. Each orthant is a hypercube that can be triangulated into $n!$ simplices (one for each permutation of the coordinates~$[n]$), as described in Section~\ref{sec:preliminaries}. We denote these simplices for the $i$-th orthant by $U_i^1(\vec{p}),\ldots,U_i^{n!}(\vec{p})$.

\begin{definition}
\label{def:super-queries}
A \emph{super query} at $\pb \in \Zn$ is a collection $\sq(\pb)$ of representative prices from the interior of each $U_i^j(\vec{p})$, where $i \in [2^n]$ and $j \in [n!]$.
\end{definition}

With a slight abuse of notation, we say that we `super query' a price vector $\pb$ if we query all price vectors in $SQ(\pb)$. Figure~\ref{fig:pos-neg-superquery} illustrates super queries in the case of two goods $(n=2)$. The following lemma demonstrates the use of super queries.

\begin{lemma}
\label{lemma:super-query}
Querying the points in $\sq(\pb)$ once is sufficient to ascertain $D_\bids(\pb')$ for any~$\pb'$ with $\| \pb - \pb' \|_\infty < 1$. With this information, we can learn all facets of the LIP containing~$\vec{p}$ and establish the existence and weight of a bid at $\pb$.
\end{lemma}
\begin{proof}
Suppose that $\vec{p}'$ is such that $\|\vec{p} - \vec{p}'\|_\infty < 1$. From construction, it follows that $\vec{p}' \in U_i^j(\vec{p})$ for some $U_i^j(\vec{p})$. However, we also know that facets in the polyhedral complex resulting from $D_\mathcal{B}$ can only be normal to a unit vector~$\vec{e}^i$, for some $i \in [n]$, or a vector of the form $\vec{e}^i - \vec{e}^j$, for some $i,j \in [n]$. As a consequence, any price in the interior of $U_i^j(\vec{p})$ demands the same unique bundle. If $\vec{p}'$ also lies in the interior of $U_i^j(\vec{p})$, then we know the bundle it demands and we are done. On the other hand, if $\vec{p}'$ is on the boundary $U_i^j(\vec{p})$, then it could be marginal, but if this is the case, from construction, we know all bundles demanded at non-marginal prices neighbouring $\vec{p}'$, as they must lie in the interior of other $U_{i'}^{j'}(\vec{p})$. This in turn implies that we can infer all bundles demanded at $\vec{p}'$ as desired.
Moreover, it follows immediately that $SQ(\vec{p})$ holds all information regarding facets containing $\vec{p}$, as facets can only take specific orientations in a SS demand correspondence. Finally, as have complete information about demand in the unit ball around $\pb$, we can simulate a bid existence query as defined in Section~\ref{sec:existence-queries} at $\pb$ without making any further queries to the demand oracle. This tells us whether there exists a bid at $\pb$ and, if so, what its weight is.

\end{proof}

\subsection{Finding a Separating Hyperplane}
\label{sec:hyperplane-finding}
Suppose $\bm{0} \leq \qb, \qb' \leq (M+1)\eb^{[n]}$ are distinct price vectors that lie in the interiors of different demand regions. Note that $\|\qb - \qb'\|_\infty \leq M+1$. As demand regions are convex and have piecewise-linear boundaries, there exists some facet~$F$ of the LIP separating $\qb$ and $\qb'$. In order to find the hyperplane containing~$F$, we first perform $O (\log M)$ steps of binary search on $\conv(\qb, \qb')$ to obtain a pair of points $\vec{s}$ and $\vec{s}'$ on either side of $F$ with $\|\vec{s} - \vec{s}' \|_\infty \leq 1/4$. By the geometry of SS demand, we know that $F$ must be normal to a vector $\eb^{i}-\eb^{j}$, where $i,j \in [n]_0$ and we denote $\eb^0 = 0$ for convenience.

Suppose we know that $F$ is normal to $\eb^i - \eb^j$ for some fixed $i$ and $j$. Then we can determine the point $\pb = \lambda \vec{s} + (1-\lambda) \vec{s}'$ ($0 \leq \lambda \leq 1$) at which the line segment $\overline{\vec{s}\vec{s}'}$ intersects~$F$, as this is obtained for the unique value $0 \leq \lambda \leq 1$ such that $\pb$ satisfies $p_i - p_j \in \mathbb{Z}$. Next, we round up $\pb$ to the nearest integer and claim that $\lceil \pb \rceil$ also lies in $F$. Indeed, we see that if a bid $\bid$ demands good $i$ at $\pb$, then it also demands $i$ at $\lceil \pb \rceil$:
\[
    b_j - \lceil p_j \rceil - 1 \leq b_j - p_j - 1 \leq b_i - p_i - 1 < b_i - \lceil p_i \rceil.
\]
Here the second inequality follows from the fact that $\bid$ demands good $i$ at $\pb$. As $\bid$ and $\lceil \pb \rceil$ are integral, this chain of inequalities implies $b_i - \lceil p_i \rceil \geq b_j - \lceil p_j \rceil$. Hence, as this holds for every bid, we see that $\pb$ and $\lceil \pb \rceil$ lie in the same facet. We perform a super query at $\lceil \pb \rceil$ in order to find $F$ (as well any other adjacent facets that $\pb$ might lie in).

\begin{algorithm}[t!]
\caption{Finding a Separating Hyperplane}
\label{alg:separating-hyperplane}
\begin{algorithmic}[1]
    \INPUT Two hyperplane witnesses $\qb$ and $\qb'$ that are separated by a hyperplane.
	\STATE Binary search between $\qb$ and $\qb'$ to find a pair of points $\vec{s}, \vec{s}'$ with distance $\|\vec{s} - \vec{s}'\|_\infty < 1/4$ and distinct demands.
    \FORALL {$(i,j) \in [n]^2_0$}
        \STATE Suppose $\vec{s}$ and $\vec{s}'$ are separated by a facet $F$ normal to $\eb^i - \eb^j$.
    	\STATE Find the unique intersection point $\pb$ of $F$ with the line segment from $\vec{s}$ to $\vec{s'}$.
    	\STATE Super query at $\lceil \pb \rceil$.
    	\IF{the super query learns a facet containing $\lceil \pb \rceil$}
    	    \RETURN the hyperplane containing this facet.
    	\ENDIF
    \ENDFOR
\end{algorithmic}
\end{algorithm}

This suggests the following algorithm for finding a separating hyperplane, specified in Algorithm~\ref{alg:separating-hyperplane}. First we perform binary search between $\qb$ and $\qb'$ in order to find points $\vec{s}$ and $\vec{s}'$ that are $1/4$ apart (w.r.t. the $L_\infty$-norm). Then, for every possible pair of goods $(i,j) \in [n]_0^2$, we suppose that $\vec{s}$ and $\vec{s}'$ are separated by a facet that is normal to $\eb^i-\eb^j$ and perform a super query at prices $\lceil \pb \rceil$ obtained as described above. As soon as such a super query finds a facet containing $\lceil \pb \rceil$, we return the hyperplane containing this facet. Overall, finding a separating hyperplane between $\qb$ and $\qb'$ in this way costs $O(\log M + n^2 2^n n!)$ queries, as each super query costs $O(2^n n!)$. Moreover, our procedure is guaranteed to find a separating hyperplane, as it exhaustively checks all possible orientations for the separating facet $F$ between $\vec{s}$ and $\vec{s}'$.

\subsection{The Main Algorithm}\label{sec:positive-negative-main}
\begin{figure}
    \centering
    \includegraphics[scale=0.7]{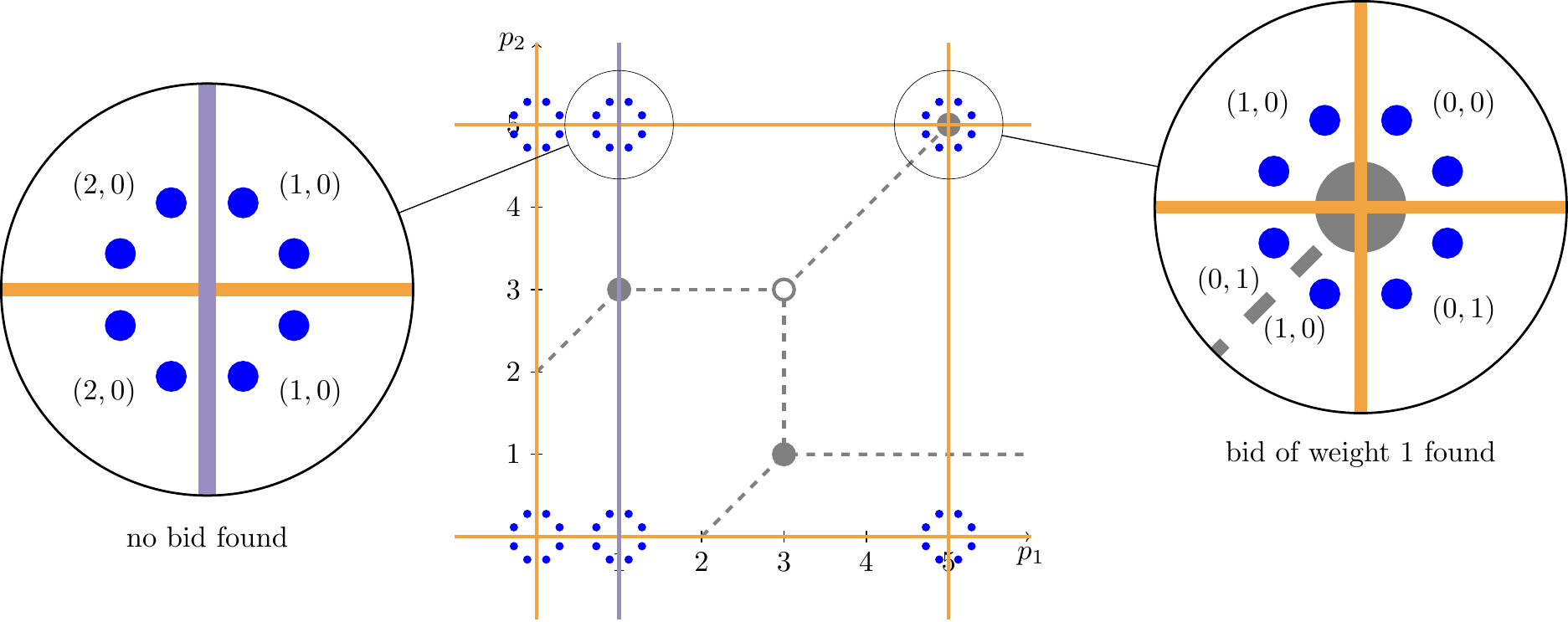}
    \caption{A snapshot of Algorithm~\ref{alg:pos-neg} learning the bid list consisting of one negative bid of weight $-1$ at $(3,3)$ and three positive bids of weight $1$ at $(1,3), (3,1)$ and $(5,5)$. The initial axis-aligned hyperplanes are drawn in orange. Moreover, the algorithm has learnt one additional vertical hyperplane $(x = 1)$ drawn in purple. The super queries made at the hyperplane intersections are represented by blue query points. Two super queries are highlighted to illustrate how the existence of a bid is determined.}
    \label{fig:pos-neg-superquery}
\end{figure}

Algorithm~\ref{alg:pos-neg} learns bid lists that may comprise positive and negative bids. The algorithm maintains a set of hyperplanes~$\mathcal{H}$ that it has learnt. We initialise~$\mathcal{H}$ with the axis-aligned hyperplanes $\eb^i = 0$ and $\eb^i = M$ for all $i \in [n]$. The algorithm also keeps track of the corresponding set of vertices $\mathcal{V}$ arising from intersections of hyperplanes in $\mathcal{H}$, the set of query points $\mathcal{Q} = \bigcup_{\vec{v} \in \mathcal{V}} \sq(\vec{v})$, as well as the set of polytopes $\mathcal{P}$ of the subdivision of $\Rn$ by the hyperplanes in~$\mathcal{H}$. Finally, $\hat{\bids}$ denotes the set of bids that the algorithm has learnt. Two query points $\qb, \qb' \in \mathcal{Q}$ are \emph{hyperplane witnesses} if they lie in the same polytope $P \in \mathcal{P}$ but have different demand. Figure~\ref{fig:pos-neg-superquery} shows a snapshot of Algorithm~\ref{alg:pos-neg} after learning a single hyperplane.

\begin{algorithm}[b!]
\caption{Learning Positive and Negative Bids}
\label{alg:pos-neg}
\begin{algorithmic}[1]
	\INITIALIZATION
	\STATE Let $\mathcal{H}$ be hyperplanes of the form $\eb^i \cdot \xb = 0$ and $\eb^i \cdot \xb = M$ for all $i \in [n]$.
	\STATE Update $\mathcal{V}, \mathcal{Q}$ and $\mathcal{P}$ according to $\mathcal{H}$.
    \STATE Super query at each $\vec{v} \in \mathcal{V}$ to check for bid at $\vec{v}$ and add each new-found bid to $\hat{\bids}$.
    \label{step:superquery-init}
    \item[\textbf{Main loop:}]
    \WHILE{$\exists$ hyperplane witnesses $\qb, \qb' \in \mathcal{Q}$}
        \STATE Use Algorithm~\ref{alg:separating-hyperplane} to find a hyperplane separating $\qb$ and $\qb'$ and add it to $\mathcal{H}$\label{step:hyperplane-finding}.
        \STATE Update $\mathcal{V}, \mathcal{Q}$ and $\mathcal{P}$ accordingly.
        \FORALL {new vertices $\vec{v}$ in $\mathcal{V}$}
            \STATE Super query at $\vec{v}$ to check for bid at $\vec{v}$ and add a newly found bid to $\hat{\bids}$.
            \label{step:superquery-loop}
        \ENDFOR
    \ENDWHILE
    \RETURN $\hat{\bids}$
\end{algorithmic}
\end{algorithm}

We now argue that Algorithm~\ref{alg:pos-neg} is well-defined and learns a bid list in $O(Bn^2)$ iterations by identifying all the hyperplanes containing a facet of the LIP. As each bid gives rise to at most $O(n^2)$ facets, the total number of such hyperplanes is $O(Bn^2)$. The algorithm learns a new hyperplane in each iteration, as hyperplane witnesses lie in the same polytope by definition, which implies that Step~\ref{step:hyperplane-finding} of the algorithm finds a new hyperplane that is not in~$\mathcal{H}$. Moreover, every bid must lie at the intersection of $n$ such hyperplanes, and we perform a super query at each intersection to check for the existence of a bid at that point. All that remains is to show that Algorithm~\ref{alg:pos-neg} does not terminate until all $O(Bn^2)$ hyperplanes containing facets of the LIP have been learnt, as this immediately implies that the algorithm identifies the locations and weights of all bids.

\begin{lemma}
\label{lemma:pos-neg-alg-correctness}
Algorithm~\ref{alg:pos-neg} learns all $O(Bn^2)$ hyperplanes containing a facet of the LIP.
\end{lemma}
\begin{proof}
    It suffices to show that there always exists a pair of hyperplane witnesses if we have not learnt all hyperplanes. Suppose $F$ is a facet of the LIP that is not contained in any hyperplane in $\mathcal{H}$. Then $F$ separates two neighbouring demand regions. Moreover, by assumption there is a polytope $P \in \mathcal{P}$ that intersects both these regions, hence there exist two non-marginal points $\pb, \pb' \in P$ at which demand differs.

    Next, recall that we perform a super query at every vertex of the subdivision of $\Rn$ by the hyperplanes in $\mathcal{H}$. Hence, for every polytope $P \in \mathcal{P}$ we have a query point close to its vertices. Suppose all these query points in $P$ demand the same bundle $\xb$. This implies that~$\xb$ is also demanded at all vertices of $P$. By convexity of demand, $\xb$ is then uniquely demanded at \emph{any} non-marginal point of the polytope. But this contradicts the fact that demand differs at $\pb$ and $\pb'$. It follows that at least two of the query points close the vertices of $P$ have distinct demand.
\end{proof}

Theorem~\ref{thm:pos-neg-query-complexity} gives the overall query complexity of learning a bid list with Algorithm~\ref{alg:pos-neg}.

\begin{theorem}
\label{thm:pos-neg-query-complexity}
	Algorithm~\ref{alg:pos-neg} requires $O\left( Bn^2 \log M + 2^n n! \binom{Bn^2}{n} \right )$ queries to learn a bid list that may consist of positive and negative bids. For $n$ constant, this is $O\left (B \log M + B^n \right )$.
\end{theorem}
\begin{proof}
	Let $H, V$ and $Q$ be the number of hyperplanes, vertices and query points in the sets $\mathcal{H}, \mathcal{V}$ and $\mathcal{Q}$ at the conclusion of Algorithm~\ref{alg:pos-neg}. Note that Algorithm~\ref{alg:pos-neg} only makes demand queries in Steps~\ref{step:superquery-init},~\ref{step:hyperplane-finding} and~\ref{step:superquery-loop}. We first count the number of queries that arise from checking for existence of a bid at vertex locations (Steps~\ref{step:superquery-init} and \ref{step:superquery-loop}).
	$\mathcal{H}$ is initialised with $2n$ axis-aligned hyperplanes and learns an additional $O(Bn^2)$ hyperplanes (cf.~Lemma~\ref{lemma:pos-neg-alg-correctness}). Hence $H = O(Bn^2)$. As every vertex in~$\mathcal{V}$ arises as the intersection of $n$ distinct hyperplanes, we have $V = O\left(\binom{Bn^2}{n}\right)$. The algorithm performs a super query of cost $2^n n!$ at each vertex in $\mathcal{V}$, leading to a total cost of
	$
	\label{eq:bid-existence-term}
	O\left (2^n n!\binom{Bn^2}{n} \right ).
	$
    Next we count the number of queries required to find new hyperplanes (Step~\ref{step:hyperplane-finding}). Section~\ref{sec:hyperplane-finding} tells us that finding a single hyperplane costs $O(\log M + n^2 2^n n!)$ queries. As the algorithm learns $O(Bn^2)$ hyperplanes, the aggregate number of queries performed by Step~\ref{step:hyperplane-finding} is
    $
    \label{eq:hyperplane-finding-term}
        O\left( B n^2 \left (\log M + n^2 2^n n! \right) \right).
	$
	Summing these two cost terms for super queries and hyperplane finding
	gives the desired query complexity for Algorithm~\ref{alg:pos-neg}.
\end{proof}

\section{Lower Bounds}\label{sec:lowerbounds}
We describe lower bounds for the complexity of learning bid lists. First we show that learning $B$ positive bids requires $\Omega(B \log M)$ queries. Our other lower bounds apply in the setting where the bid list may comprise positive and negative bids and make use of a carefully constructed `island gadget' that consists of $2^n$ positive and $2^n$ negative bids of unit weight. For bid lists $\bids$ that need not be valid, the island gadget immediately implies that $\Omega \left( \left (\frac{M}{4} \right )^n \right)$ queries are required to learn $\bids$ (roughly, price-space has to be queried exhaustively). For the case with valid bid lists, we construct an adversarial game to obtain a lower bound of $\Omega\left (\left (({B-2^{n+1}})/{8n^2} \right )^n \right )$ on the query complexity. We see that $n$ must be held constant for the query complexity to be polynomial. In this regime, our lower bounds and the upper bounds from Section~\ref{sec:generalbids} imply a query complexity of $\Theta \left (B \log M + B^n \right )$ for constant~$n$.

\begin{theorem}\label{thm:simple-lower-bound}
	Any algorithm for learning bid lists requires $\Omega(B \log M) $ queries.
\end{theorem}
\begin{proof}
Suppose an adversary places $B$ positive bids at $B$ integral points in the space $\{ (x_1, 0, \ldots, 0) \mid x_1 \in [M]_0 \})$. By a standard decision tree argument, any algorithm must make $\Omega(B \log M)$ queries to learn the location of the bids.
\end{proof}

\subsection{The Island Gadget}\label{sec:island-gadget}
We now introduce the \emph{island gadget}, which allows us to locally change the demand correspondence without affecting demand outside the convex hull of the gadget bids. This is illustrated in Figure~\ref{fig:gadget}.
Lemma~\ref{lemma:gadget-doesnt-leak} establishes that the gadget only influences demand locally. Let $\rho(\vec{x})$ denote the weight function that assigns positive weight~$1$ to bids with an even number of odd entries and negative weight~$-1$ otherwise.

\begin{definition}\label{def:gadget}
    The gadget $G$ at position $\bf{0}$ consists of the following $2^{n+1}$ positive and negative unit bids that sit on the vertices of two unit hyper-cubes. The bids on the first hyper-cube lie at $\bid \in \{0,1\}^n$ and have weight $\rho(\bid)$. The bids on the second hyper-cube lie at $\bid \in \{2,3\}^n$ and have weight $-\rho(\bid)$. In order to place $G$ at position $\vec{x}$, we add $\vec{x}$ to the position of each bid (without changing the weights).
\end{definition}

\begin{figure}
    \centering
    \subfigure[]{\includegraphics[scale=1]{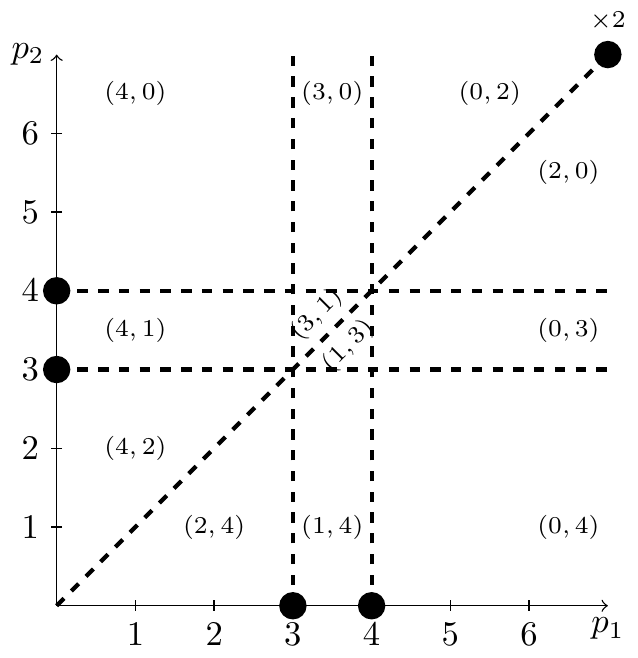}}
    \subfigure[]{\includegraphics[scale=1]{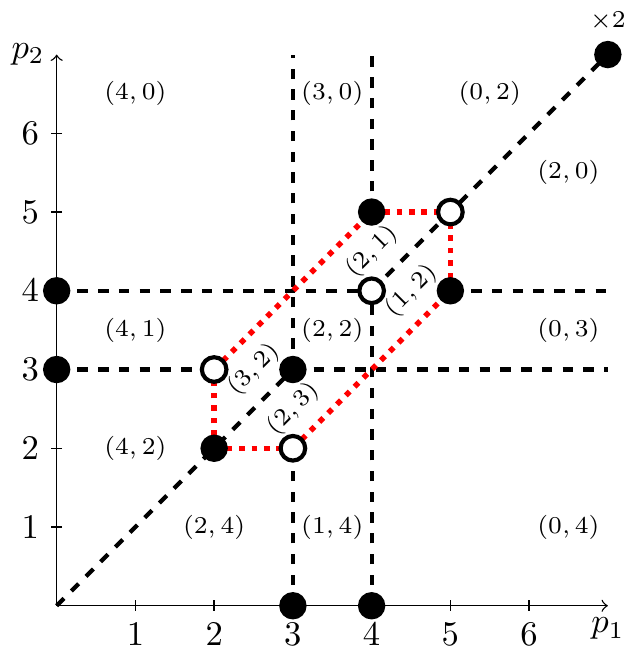}}
    \caption{A demand correspondence with (right) and without (left) an \emph{island gadget} placed at~$(2,2)$.
    The gadget consists of the eight bids contained in the convex region denoted by the red dotted lines. Note that demand differs only at prices \emph{inside} inside this region.}
    \label{fig:gadget}
\end{figure}

\begin{lemma}
\label{lemma:gadget-doesnt-leak}
    The bids of gadget $G$ placed at position~$\vec{x}$ demand nothing at prices $\vec{p} \not \in \xb + [3]_0^n$.
\end{lemma}

The proof of Lemma~\ref{lemma:gadget-doesnt-leak} makes use of the following technical result.

\begin{lemma}\label{lemma:parity-count}
    Let $\vec{c} \in \R^n$ be a constraint vector. If $c_i \geq 1$ for some $i \in [n]$, then the number of vectors with an even number of 1 entries in $\{\vec{x} \in \{ 0,1 \}^n \mid \vec{x} \leq \vec{c} \}$ is equal to the number of vectors with an odd number of 1 entries.
\end{lemma}
\begin{proof}
    Induction on $n$.
\end{proof}

\begin{proof}[Proof of Lemma~\ref{lemma:gadget-doesnt-leak}]
    Recall that gadget $G$ at position $\xb$ consists of the bids $\bid \in \xb + \{0,1\}^n$ of weight $\rho(\bid - \xb)$ and the bids $\bid \in \xb+\{2,3\}^n$ of weight $-\rho(\bid - \xb)$. In order to prove that the bids of $G$ do not aggregately demand items of any goods at prices $\vec{p} \not \in \xb + [3]_0^n$, we show for every good $i$ and every such price that the number of negative and positive gadget bids demanding $i$ at $\vec{p}$ is the same.

    Fix a good $i \in [n]$ and prices $\vec{p} \not \in
		\xb + [0,3]^n$. Without loss of generality, we assume that $\vec{p}$ is non-marginal and $p_i \not = p_j$ for all $i,j \in [n]$. Recall that an integral bid $\bid$ uniquely demands good $i$ at $\vec{p}$ if and only if we have $b_i > p_i$ and $b_i - p_i > b_j - p_j$ for all $j \in [n] \setminus \{ i \}$. The last condition can be rewritten as $b_i - b_j > p_i - p_j$. Hence we can express the set of bids in the bottom and top cube that (uniquely) demand a positive or negative item of good $i$ as
    \begin{align}
        \label{eq:bottom-cube}
        \{ \bid \in \xb + \{0, 1\}^n &\mid b_i > p_i \text{ and } b_i - b_j > p_i - p_j,\ \forall j \in [n] \setminus \{i\} \} \\
        \label{eq:top-cube}
        \text{and }\{ \bid \in \xb + \{2, 3\}^n &\mid b_i > p_i \text{ and } b_i - b_j > p_i - p_j,\ \forall j \in [n] \setminus \{i\} \},
    \end{align}
    respectively. We show by case distinction on the possible value of $p_i$ that the number of positive and negative bids in \eqref{eq:bottom-cube} and \eqref{eq:top-cube} is equal.

    \textbf{Case I:}
    Suppose first that $p_i \not \in \xb + [3]_0^n$. If $p_i > x_i+3$, the two sets \eqref{eq:bottom-cube} and~\eqref{eq:top-cube} are empty and we are done. Hence assume that $p_i < x_i$. Then we can express \eqref{eq:bottom-cube} and~\eqref{eq:top-cube} as
    \begin{align*}
        &\xb + \{ \vec{v} \in \{0,1\}^n \mid v_i - v_j > p_i - p_j,\ \forall j \in [n] \setminus \{i\} \} \\
    \text{and } &\xb + 2\eb^{[n]} + \{ \vec{v} \in \{0,1\}^n \mid v_i - v_j > p_i - p_j,\ \forall j \in [n] \setminus \{i\} \}.
    \end{align*}

    We see that there is a one-to-one correspondence $\bid \rightarrow \bid + 2\eb^{[n]}$ between bids in \eqref{eq:bottom-cube} and \eqref{eq:top-cube} that preserves the number of odd entries in the bid vectors. A bid $\bid$ in the bottom cube is positive if and only if $\bid - \xb$ has an even number of odd entries and the reverse is true for the top cube. Hence the total number of positive and negative bids in \eqref{eq:bottom-cube} and \eqref{eq:top-cube} is equal.

    \textbf{Case II:}
    Suppose that $p_i \in [x_i,x_i+3]$. Then there exists a good $k \in [n] \setminus \{ i \}$ with $p_k < x_k$ or $p_k > x_k+3$. We proceed by further case distinction. Note that $b_i - b_j \in \{-1, 0, 1 \}$ for all gadget bids $\bid$.

    \textbf{Case II.1:}
    Suppose $p_i - p_j \geq 1$ for some $j \in [n] \setminus \{ i \}$. Then the sets \eqref{eq:bottom-cube} and \eqref{eq:top-cube} are empty, as the constraint $b_i - b_j > p_i - p_j$ is violated.

    \textbf{Case II.2:} Now suppose $p_i - p_j < 1$ for every $j \in [n] \setminus \{ i \}$ and we have a non-empty set $K$ of indices $k$ for which $0 < p_i - p_k < 1$.
    Then for all bids $\bid$ in \eqref{eq:bottom-cube} and \eqref{eq:top-cube} we have $b_i - b_k = 1$ for all $k \in K$. Hence $b_i$ and $b_k$ can only take one value.

    If $K = [n] \setminus \{ i \}$, this implies that \eqref{eq:bottom-cube} and \eqref{eq:top-cube} both contain a single bid of odd parity and we are done. If $K \subset [n] \setminus \{ i \}$, then for $j \not \in K$ we have $p_i - p_j < 0$ and the constraint $b_i - b_j > p_i - p_j$ holds vacuously, so $b_j$ can take any value. We apply Lemma~\ref{lemma:parity-count} to see that \eqref{eq:bottom-cube} and \eqref{eq:top-cube} each contain an equal number of positive and negative bids.

    \textbf{Case II.3:} Suppose $p_j > p_i$ for all $j \in [n] \setminus \{ i \}$. Then there exists a good $k \in [n] \setminus \{ i \}$ such that $p_k > x+3$. We consider the two subsets $\{b \in \eqref{eq:bottom-cube} \mid b_i = x_i \}$ and $\{b \in \eqref{eq:bottom-cube} \mid b_i = x_i+1 \}$ of \eqref{eq:bottom-cube} separately. If $\bid$ is a bid in the first subset, then $b_i = x_i$ and $b_i > p_i$ imply that $p_i - p_k < -1$, so there is no constraint on the value of $b_k$. If $\bid$ is a bid in the second subset, then there is no constraint on the value of $b_k$. Hence in both cases each subset is either empty or contains an equal number of positive and negative bids by Lemma~\ref{lemma:parity-count}.
\end{proof}

\subsection{A Lower Bound for Positive and Negative Bids}
\label{sec:pos-neg-lower-bound}
From now on, we assume that the unknown bid list may comprise positive and negative bids. We primarily consider the setting where the bid list is assumed to be valid and specify an adversarial game where the adversary can force the player to make $\Omega \left (\left (\frac{B-2^{n+1}}{8n^2} \right )^n \right )$ queries. Note that this yields a query complexity of $\Omega(B^n)$ if $n$ is constant. A similar but simpler adversarial game is then applied to obtain a lower bound for the setting where the bid list may be invalid. The resulting lower bounds for bid lists comprising positive and negative bids are stated in Theorem~\ref{thm:pos-neg-lower-bound}.

Fix a parameter $k \in \mathbb{N}$ and let $M = 4k$. The adversary positions the island gadget consisting of $2^{n+1}$ bids at exactly one of $k^n$ possible points $\vec{x}$  of the lattice $4[k-1]_0^n$. He also places \emph{boundary bids} of weight 1 on the boundary of the $M$-cube as follows. For every good $i \in [n]$, there are $M$ positive bids at all points $m\eb^i$ with $m \in [M-1]_0$, and for every pair of goods $(i,j) \in [n]^2$ with $i \not = j$, there are $M$ positive bids at all points $m\eb^i + M\eb^j$ with $m \in [M-1]_0$. Hence, in total, the adversary creates a bid list with $B = 2^{n+1} + 4k \left (n + \binom{n}{2} \right )$ positive and negative bids. The player wishes to identify where the adversary placed the gadget. Lemma~\ref{lemma:gadget-doesnt-leak} shows that placing the gadget at $\vec{x}$ only influences demand at prices inside the cube $\xb + [3]_0^n$. Hence the player must make queries inside at least $k^n-1$ cubes to determine where the gadget was placed. Lemma~\ref{lemma:gadget-is-covered} shows that the bid list created by the adversary is valid. This leads to the lower bounds stated in Theorem~\ref{thm:pos-neg-lower-bound}.

\begin{lemma}\label{lemma:gadget-is-covered}
    The bids placed by the adversary are valid.
\end{lemma}

Recall that $\bids^G$ denotes the bids of the island gadget. In order to prove Lemma~\ref{lemma:gadget-is-covered}, we define the subset $\bids_i^G$ of $\bids^G$ that contains all bids indifferent between a good $i$ and the reject good at some fixed price vector $\pb$. The subset of bids indifferent between two non-reject goods $i$ and $j$ is defined similarly, albeit separately for the `lower' and the `upper' cube of the gadget. Specifically, we define
\begin{equation}\label{eq:gadget-hod}
    \bids^G_i \coloneqq \{\bid \in \{0,1\}^n \cup \{2,3\}^n \mid b_i = p_i \text{ and } b_k \leq p_k, \ \forall k \in [n]\}
\end{equation}
as well as
\begin{align}\label{eq:gadget-flange}
    \bids^L_{ij} &\coloneqq \{\bid \in \{0,1\}^n \mid b_i \geq p_i, b_j = b_i - p_i + p_j \text{ and } b_k \leq b_i - p_i + p_k \ \forall k \not \in \{i,j\} \} \\
    \bids^U_{ij} &\coloneqq \{\bid \in \{2,3\}^n \mid b_i \geq p_i, b_j = b_i - p_i + p_j \text{ and } b_k \leq b_i - p_i + p_k \ \forall k \not \in \{i,j\} \}.
\end{align}

\begin{lemma}\label{lemma:hod-count}
    Let $n \geq 3$ and $\pb \in \{0,1\}^n \cup \{2,3\}^n$. The sum of the weights of the gadget bids indifferent between $i$ and the reject good $0$ $\pb$ is $-1$, $0$ or $1$.
\end{lemma}
\begin{proof}
    Suppose first that $\pb \in \{0,1\}^n$. We show that the sum of the weights of bids in $\bids^G_i$ is $0$ or $1$. Recalling the definition of $\bids^G_i$ from \eqref{eq:gadget-hod}, note that if $p_k = 0$ for all $k \not = i$, then $\bid = \vec{0}$ is the only bid in this set and its weight is positive. If we have $p_k = 1$ for some $k \not = i$, then the set contains an equal number of positive and negative bids by Lemma~\ref{lemma:parity-count}. Similarly, we can show that the sum of the weights of bids in $\bids^G_i$ is $-1$ or $0$ when $\pb \in \{2,3\}^n$.
\end{proof}

\begin{lemma}\label{lemma:flange-count}
    Let $n \geq 3$ and $\pb \in \{0,1\}^n \cup \{2,3\}^n$. If $\bids^L_{ij}$ has more than one bid, the sum of the weights of the bids indifferent between $i$ and $j$ is $0$ or $1$. Similarly, if $\bids^U_{ij}$ has more than one bid, the sum of the respective bid weights is $-1$ or $0$.
\end{lemma}
\begin{proof}
    We prove the claim for $\bids^L_{ij}$; the proof for $\bids^U_{ij}$ is identical. Note that in order for $\bids^L_{ij}$ to contain any bids, we must have $\pb \in \{0,1\}^n$.
    Suppose $p_i = 1$. Then for any bid $\bid \in \bids^L_{ij}$, the values of $b_i$ and $b_j$ are constrained to $b_i=1$ and $b_j = p_j$. Moreover, the values of $b_k$ are constrained by $b_k \leq p_k$. Hence for $\bids^L_{ij}$ to contain at least two bids, we must have $p_k = 1$ for some $k \not \in \{i,j\}$, allowing $b_k$ to take values~0 and 1. Lemma~\ref{lemma:parity-count} implies that $\bids^L_{ij}$ contains an equal number of positive and negative bids.
    Now suppose $p_i = 0$. Then $b_i$ can take value 0 or 1 and we consider the two subsets of $\bids^L_{ij}$ where $b_i$ takes each value separately. Let
    $S \coloneqq \{ \bid \in \bids^L_{ij} \mid b_i = 0 \}$
    and
    $T \coloneqq \{ \bid \in \bids^L_{ij} \mid b_i = 1 \}$.
    By definition, $\bid$ is a bid in $S$ if $b_i=0$, $b_j = p_j$ and $b_k \leq p_k$. By the same argument as above, $S$ either contains a single positive bid at $b_j = p_j$ and $b_k = 0$ for all $k \not = i$, or it contains an equal number of positive and negative bids.
    
    Finally, note that $\bid$ is a bid in $T$ if $b_i = 1$, $b_j = 1 + p_j$ and $b_k \leq 1+ p_k$ for all $k \in [n]$. The last constraint is vacuous as it holds for all possible values of $b_k$ and $p_k$. Hence $T$ is empty if $p_j = 1$ and contains an equal number of positive and negative bids otherwise. Together, these properties of $S$ and $T$ imply the claim.
\end{proof}

We are now ready to prove Lemma~\ref{lemma:gadget-is-covered}.

\begin{proof}[Proof of Lemma~\ref{lemma:gadget-is-covered}]
    Suppose that gadget $G$ is placed at $\vec{x}$. We apply Theorem~\ref{thm:valid-bid-list} to prove the claim. Moreover, we can use a result given by Lemma B.10 in \cite{BGKL19} to restrict ourselves to checking the validity criterion in Theorem~\ref{thm:valid-bid-list} only at a finite number of price points (see \cite{BGKL19}, Appendix B for details). In our case, it suffices to verify the conditions for each vertex $\vec{p} \in \xb + \{ 0, 1 \}^n$ and $\vec{p} \in \xb + \{2, 3 \}^n$.

    Fix $\pb \in \xb + \{0,1\}^n \cup \{2,3\}^n$ and distinct goods $i,j \in [n]_0$. We consider first the case where $i=0$ and $j \in [n]$. By Lemma~\ref{lemma:hod-count}, the sum of weights of the gadget bids indifferent between $0$ and $j$ is at least $-1$. Secondly, by construction of the adversarial game, there is a gadget bid at $p_i \eb^i$ and this bid is indifferent between $i$ and the reject good. Hence the total sum of weights of the bids in the game indifferent between $i$ and rejecting is at least~$0$.

    Now let $i,j$ be distinct (non-reject) goods in $[n]$. Note that $\bid \in \{0,1\}^n$ is indifferent between $i$ and $j$ at $\pb$, then $\bid' = \bid + \vec{2} \in \xb + \{2,3\}^n$ is also indifferent between $i$ and $j$. Moreover, $\bid$ and $\bid'$ have opposite weights. This fact, together with Lemma~\ref{lemma:flange-count}, implies that the sum of the weights of the gadget bids is at least $-1$. Again, note that for any $\pb$ there is a boundary bid that is indifferent between $i$ and $j$. Indeed, recall that a bid is indifferent between $i$ and $j$ if and only if $b_i - b_j = p_i - p_j$, $b_i \geq p_i$ and $b_k - p_k \leq b_i - p_i$. If $p_i \geq p_j$, then the bid $M \eb^i + (M - (p_i - p_j))\eb^j$ satisfies these conditions. Similarly, if $p_i < p_j$, then the bid $(M-(p_i - p_j))\eb^i - M\eb^j$ satisfies this condition. As this bid is present by construction of the adversarial game, we see that the total sum of weights of all bids that are indifferent between $i$ and $j$ in the game is at least $0$ in all cases.
\end{proof}

\begin{theorem}
\label{thm:pos-neg-lower-bound}
    Let $\bids$ be a bid list that may comprise positive and negative bids. If $\bids$ is valid, any algorithm requires $\Omega \left (\left (\frac{B-2^{n+1}}{8n^2} \right )^n \right )$ queries to learn $\bids$. If $\bids$ is allowed to be invalid, $\Omega \left( (\frac{M}{4} )^n \right)$ queries are required.
\end{theorem}
\begin{proof}
Consider first the case where $\bids$ is valid. Let $\bids$ be the adversary's bid list, which is valid by Lemma~\ref{lemma:gadget-is-covered}. By Lemma~\ref{lemma:gadget-doesnt-leak}, we know that the gadget placed at $\xb$ does not affect demand outside the cube $\xb + [3]_0^n$. Hence by a standard decision tree argument, any algorithm requires $k^n-1$ queries to learn the location of the gadget $G$.
Note that the adversary places a total of $B \leq 2^{n+1} + 8kn^2$ gadget and boundary bids. Solving for~$k$ gives $k \geq \frac{B-2^{n+1}}{8n^2}$. Hence expressing the query complexity of $k^n-1$ in terms of $B$ and $n$ yields the desired expression.

For the case where $\bids$ may be invalid, we construct an even simpler adversarial game, where the adversary positions the island gadget as above but places no boundary bids. By the same argument, learning the location of the gadget incurs at least $\Omega \left( (\frac{M}{4} )^n \right)$ queries.
\end{proof}

Using the fact that $\max\{f,g\} = \Omega(f+g)$ for any positive functions $f$ and $g$, we can combine the lower bounds from Theorem~\ref{thm:simple-lower-bound} and Theorem~\ref{thm:pos-neg-lower-bound}, along with the upper bound from Algorithm~\ref{alg:pos-neg} in Section~\ref{sec:generalbids}, to identify the query complexity of expressing SS demand by means of positive and negative bids.

\begin{corollary}\label{cor:lower-bound}
    For constant $n$, learning SS demand requires $\Theta(B^n + B \log M)$ queries.
\end{corollary}

\section{Conclusions}
Our algorithms for learning demand are conceptually simple and provide the first systematic approach for bidders to express their preferences in the bidding language used by the Product-Mix Auction. This allows bidders with non-technical backgrounds to participate in these auctions under the mild assumption that they are able to answer demand oracle queries. In the setting where demand can be expressed using positive bids only, our algorithm achieves linear query complexity. When demand may only be expressible using positive \emph{and} negative bids, our hyperplane finding algorithm performs well if the number of goods is not too large.
Further work could address extending our positive-bid algorithm to allow a small number of negative bids, approximate learning of the demand function, and dealing with errors in answers to queries.

\begin{acks}
    The authors thank their anonymous reviewers for insightful comments and suggestions. We are also grateful for valuable comments from Elizabeth Baldwin.
\end{acks}

\bibliographystyle{ACM-Reference-Format}
\bibliography{refs.bib}

\end{document}